\documentclass[11pt]{amsart}

\usepackage[letterpaper,margin=1in]{geometry}
\usepackage{amssymb,amsmath,amsfonts,stmaryrd}
\usepackage[utf8]{inputenc}     
\usepackage{dsfont}            
\renewcommand\mathbb\mathds     

\usepackage{graphicx}
\usepackage{tikz}
\usetikzlibrary{cd,calc,shapes.geometric,decorations.markings,decorations.pathreplacing,decorations.pathmorphing,arrows}
\usepackage{tikz-cd}
\usepackage{booktabs}
\usepackage{xcolor}
\usepackage{url}
\usepackage{adjustbox}
\usepackage{placeins}
\usepackage{arydshln}
\usepackage{comment}
\usepackage{marginnote}
\usepackage{booktabs}
\usepackage{array}
\usepackage{lipsum}   
\usepackage{bbold}
 \usepackage{booktabs}
 \usepackage{tabularx}
 \usepackage{array}

\usepackage[colorlinks,
            linkcolor=blue,
            citecolor=green!70!black,
            pdfproducer={pdfLaTeX},
            pdfpagemode=None,
            bookmarksopen=true,
            bookmarksnumbered=true]{hyperref}

\newtheorem{theorem}{Theorem}[section]

\newtheorem{proposition}[theorem]{Proposition}
\newtheorem{lemma}[theorem]{Lemma}
\newtheorem{corollary}[theorem]{Corollary}

\theoremstyle{definition}
\newtheorem{definition}[theorem]{Definition}
\newtheorem{example}[theorem]{Example}
\newtheorem{construction}[theorem]{Construction}
\newtheorem{notation}[theorem]{Notation}
\newtheorem{ansatz}[theorem]{Ansätz}

\newtheorem{conjecture}[theorem]{Conjecture}

\theoremstyle{remark}
\newtheorem{rem}[theorem]{Remark}

\usepackage{cleveref}

\crefname{theorem}{theorem}{theorems}
\Crefname{theorem}{Theorem}{Theorems}

\crefname{proposition}{proposition}{propositions}
\Crefname{proposition}{Proposition}{Propositions}

\crefname{lemma}{lemma}{lemmas}
\Crefname{lemma}{Lemma}{Lemmas}

\crefname{corollary}{corollary}{corollaries}
\Crefname{corollary}{Corollary}{Corollaries}

\crefname{definition}{definition}{definitions}
\Crefname{definition}{Definition}{Definitions}

\crefname{example}{example}{examples}
\Crefname{example}{Example}{Examples}

\crefname{remark}{remark}{remarks}
\Crefname{remark}{Remark}{Remarks}

\crefname{conjecture}{conjecture}{conjectures}
\Crefname{conjecture}{Conjecture}{Conjectures}

\crefname{maintheorem}{main theorem}{main theorems}
\Crefname{maintheorem}{Main Theorem}{Main Theorems}

\usepackage{cleveref}

\crefname{maintheorem}{main theorem}{main theorems}
\Crefname{maintheorem}{Main Theorem}{Main Theorems}

\crefname{dummy}{section theorem}{section theorems} 
\Crefname{dummy}{Section Theorem}{Section Theorems}

\crefname{theorem}{theorem}{theorems}
\Crefname{theorem}{Theorem}{Theorems}

\crefname{lemma}{lemma}{lemmas}
\Crefname{lemma}{Lemma}{Lemmas}

\crefname{proposition}{proposition}{propositions}
\Crefname{proposition}{Proposition}{Propositions}

\crefname{corollary}{corollary}{corollaries}
\Crefname{corollary}{Corollary}{Corollaries}

\crefname{definition}{definition}{definitions}
\Crefname{definition}{Definition}{Definitions}

\crefname{example}{example}{examples}
\Crefname{example}{Example}{Examples}

\crefname{remark}{remark}{remarks}
\Crefname{remark}{Remark}{Remarks}

\crefname{conjecture}{conjecture}{conjectures}
\Crefname{conjecture}{Conjecture}{Conjectures}

\usepackage{dsfont}
\renewcommand\mathbb\mathds

\newcommand\bC{\mathbb C}

\newcommand\Z{\mathbb Z}

\newcommand{\C}{\mathbf C}
\newcommand{\Q}{\mathbf Q}

\newcommand\cA{\mathcal A}

\newcommand\cC{\mathcal C}
\newcommand\cD{\mathcal D}
\newcommand\cE{\mathcal E}
\newcommand\cF{\mathcal F}

\newcommand\cP{\mathcal P}
\newcommand\cQ{\mathcal Q}

\newcommand\cS{\mathcal S}
\newcommand\cT{\mathcal T}
\newcommand\cU{\mathcal U}
\newcommand\cV{\mathcal V}
\newcommand\cW{\mathcal W}

\newcommand\rB{\mathrm B}

\newcommand\rH{\mathrm H}

\newcommand\rK{\mathrm K}
\newcommand\rL{\mathrm L}

\newcommand\Sp{\mathbf {Sp}}

\newcommand\pt{\mathrm{pt}}
\usepackage[mathscr]{euscript}

\newcommand\fC{\mathfrak C}

\newcommand\bbD{\mathbf D}

\newcommand{\id}{\mathrm{id}}

\newcommand{\loc}{\mathrm{loc}}
\newcommand{\Spin}{\mathrm{Spin}}

\newcommand{\op}{\mathrm{op}}

\newcommand{\Fib}{\mathrm{Fib}}
\newcommand{\coFib}{\mathrm{coFib}}

\newcommand\longto\longrightarrow
\newcommand\mono\hookrightarrow
\newcommand\epi\twoheadrightarrow

\newcommand\<\langle
\renewcommand\>\rangle
\newcommand\sminus\smallsetminus

\DeclareMathOperator\End{End}
\DeclareMathOperator\Ext{Ext}

\DeclareMathOperator\Aut{Aut}

\newcommand\Vect{\mathbf{Vect}}

\newcommand\Mod{\cat{Mod}}

\newcommand{\Bimod}{\mathbf{Bimod}}

\newcommand{\Fun}{\mathrm{Fun}}

\newcommand{\Mat}{\mathrm{Mat}}

\let\lim\relax
\DeclareMathOperator{\lim}{lim}

\newcommand{\Hom}{\mathrm{Hom}}

\newcommand{\Witt}{\mathcal{W}itt}
\newcommand{\sWitt}{s\mathcal{W}}
\newcommand{\TS}{\mathbf{TS}}
\newcommand{\perf}{\mathrm{perf}}
\newcommand{\rMor}{\mathrm{Mor}}

\newcommand\cat[1]{\mathbf{#1}}

\title{{The Classification of Pauli Stabilizer Codes: \\A Lattice and Continuum Treatise
}}
\begin{document}
\author{Bowen Yang}
\address{\textnormal{\url{bowen_yang@g.harvard.edu}}, Harvard University}
\author{Matthew Yu}
\address{\textnormal{\url{yum@maths.ox.ac.uk}}, Oxford University}

\begin{abstract}
   We classify mobile Pauli stabilizer codes up to gapped interfaces and coarse-graining using the framework of algebraic $\rL$-theory. We compare this classification with that of framed TQFTs, theories that arise naturally in the continuum, highlighting a close structural relationship between the two. Our approach is formulated in the category of perfect chain complexes equipped with quadratic  functor over the Laurent polynomial ring $R = \mathbb{Z}/p[x_1^{\pm1}, \ldots, x_n^{\pm1}]$, within which the collection of topological operators of Pauli stabilizer codes arise naturally as objects.  In particular, we establish a bulk–boundary correspondence for lattice theories: the equivalence class of a Pauli stabilizer code up to gapped interface is described by a Clifford QCA in one dimension higher. This is done using the universal target category for stabilizer codes, which is the categorical spectrum whose existence and universal properties are introduced in this work. We conclude by highlighting subtle differences between the  classification of Pauli stabilizer codes and TQFTs, leading to qualitative distinctions between lattice and continuum theories.
\end{abstract}
\maketitle

\tableofcontents

\newpage 

\section{Introduction}

Pauli stabilizer codes are quantum spin systems defined by commuting subgroups of the Pauli operator algebra on a collection of qubits. The common eigenspaces of these operators define highly entangled subspaces in which logical degrees of freedom are encoded via stabilizer constraints. This algebraic structure places stabilizer codes among the most fundamental and widely studied classes of quantum error-correcting codes. Pauli stabilizer codes also furnish a natural bridge to condensed matter physics, where they can be interpreted as exactly solvable lattice models exhibiting quasiparticle excitations. In \cite{haah2013commuting}, Pauli stabilizer codes are recast as modules over Laurent polynomial rings. The works \cite{ruba2024homological,RY} push this perspective significantly further by bringing in methods from homological algebra, with derived functors yielding key invariants for analyzing the associated topological excitations.

In the continuum, the study of quantum field theories and their renormalization group flows naturally singles out a distinguished class of theories that lack local propagating degrees of freedom and instead support only extended, topological excitations encoded in their ground-state degeneracy and long-range entanglement; these theories are known as topological quantum field theories (TQFTs). These features give rise to striking physical phenomena, including robust edge states, fractionalized excitations, and anyonic statistics, with implications for both field theory and applications such as fault-tolerant quantum computation. The axiomatization of TQFTs in terms of higher fusion categories has, in turn, driven significant progress in the study of higher fusion categories themselves \cite{Gaiotto:2019xmp,Decoppet:2023bay,Bhardwaj:2024xcx,Decoppet:2024htz,xu2024etale,Kong:2024ykr,Stockall:2025rng,Teixeira:2025qsg}, while also illuminating how the resulting mathematical framework can be fruitfully applied in physical contexts, for instance in the analysis of matching anomalies for finite symmetries with TQFTs \cite{Debray:2025kfg,Debray:2026sqw}.

One of the central aims of this paper is to establish a precise correspondence between the classification of Pauli stabilizer codes and the classification of a distinguished class of topological quantum field theories, namely framed TQFTs defined on manifolds equipped with a framing. Recent developments have brought the classification of TQFTs to a high level of maturity, with manifold-theoretic methods—particularly the machinery of surgery theory—emerging as a central organizing principle. From this perspective, it is perhaps unexpected that despite the entirely algebraic formulation of Pauli stabilizer codes, similar techniques used to classify TQFTs can be brought to bear on their classification.

A well established mathematical framework that promotes algebraic structures to objects of homotopy theory is algebraic $\rK$-theory, a powerful tool that associates to an algebraic object a spectrum whose homotopy groups encode its intrinsic invariants.  On the other hand, the study of lattice theories and lattice topological phases seeks to extract robust invariants from the algebra of observables in order to characterize such phases. In many cases, these invariants take the form of abelian groups; from the modern homotopy-theoretic perspective, this strongly suggests a refinement to spectra, and hence to the framework of generalized cohomology theories.
In the context of classifying Pauli stabilizer codes, we will instead make use of a closely related framework, namely algebraic 
$\rL$-theory. While it shares the same conceptual foundations as algebraic 
$\rK$-theory, it is tailored to settings involving categories equipped with a quadratic functor. One of the central insights carried by algebraic $\rL$-theory from manifold topology into a categorical setting is the notion of Poincaré duality, which underlies the entire framework of surgery theory. By recasting the excitations of a Pauli stabilizer code within this 
$\rL$-theoretic formalism, we place them in a setting where duality and quadratic structures are manifest, thereby enabling the application of surgery-theoretic techniques analogous to those used in the classification of TQFTs.

With classification results in hand for both theories, we can draw a more refined comparison between continuum and lattice frameworks. In particular, we identify a correspondence between Pauli stabilizer codes and Clifford QCAs in one higher dimension, which have been shown to be classified by algebraic 
$\rL$-theory \cite{haah2025topological, sun2025clifford, Yang:2025jvn}. In the continuum, the bulk determines the Morita class of the boundary category, thereby classifying boundary theories up to gapped interfaces; an analogous phenomenon emerges in the setting of stabilizer codes.
Finally, we find a sharp contrast between the continuum and lattice settings: while continuum framed invertible TQFTs admit gapped boundaries only in six dimensions \cite{JFR2025}, corresponding to the nontrivial Arf–Brown–Kervaire invariant, invertible topological lattice phases described by Pauli stabilizer codes admit gapped boundaries in all dimensions. This suggests that, outside of six dimensions, the passage to the continuum necessarily closes the gap on the boundary.

\subsection{Pauli Stabilizer Codes}\label{subsection:introtopauli}
We begin with a review of definitions. Readers already familiar with the homological formalism of Pauli stabilizer codes may proceed directly to the main results in \S\ref{subsection:mainresults}.
A Pauli stabilizer code is defined with respect to a spatial lattice $\Z^n$, and a local degree of freedom at each site given by a qudit. We allow for qudits at each site to be of dimension $p$, and assign the algebra at each site which is $m$ copies of the $p$-state clock and shift system. The generating operators are denoted by $Z_{\lambda,j}$ (clock) and $X_{\lambda,j}$ (shift), for $j= 1,\ldots m$ and $\lambda \in \Lambda$, and obey the following relations:
\begin{align}
    Z_{\lambda,j}\, X_{\lambda,j}& = e^{\frac{2\pi i}{p}} X_{\lambda,j}\, Z_{\lambda,j}\,,  \quad Z_{\lambda,j}\, X_{\lambda,k} = X_{\lambda,k}\, Z_{\lambda,j} \quad j\neq k\,,\\\notag
    & \quad \quad \quad  X^p_{\lambda,j} = Z^p_{\lambda,j} = \id\,.
\end{align}

\begin{definition}
    The \textit{Pauli group} $\cP$ consists of all operators defined on the lattice $\Z^n$ of the form
    \begin{equation}
        \phi \prod_{\lambda \in \Z^n} \prod^m_{j=1} X^{\alpha_{\lambda,j}}_{\lambda,j} Z^{\beta_{\lambda,j}}_{\lambda,j}, \quad \phi \in \mathbb C^\times \,,\, \text{and}\,\, \alpha_{\lambda,j},\beta_{\lambda,j} \in \Z/p.
    \end{equation}
    The \textit{projective Pauli group} $P$ is the quotient of $\cP$ by its center of all complex phase factors. $P$ is isomorphic to the additive group of finitely supported additive functions $\Z^n \to \Z^{2m}_p$.
\end{definition}

We note that $P$ exhibits an action by the translations along each dimension of the lattice. This means $P$ is a module for the ring $R=\Z/p[\Z^{n}]\cong \Z/p[x_1^{\pm}, \dots, x_n^{\pm}]$, consisting of Laurent polynomials in $d$-variables.
Each variable $x_i$ is a generator of translation in the $i$-th lattice direction. In particular, translation by a lattice vector $\lambda = (\lambda_1,\ldots,\lambda_n)\in \Z^n$ corresponds to  multiplication by the monomial 
$$x^\lambda:=\prod^n_{i=1} x^{\lambda_i}_i.$$ 

\begin{definition}
   The ring $R=\Z/p[\Z^{n}]$ can be equipped with an involution:
\begin{equation}
    \overline{\bullet}: R \rightarrow R\,,
\end{equation}
given by $\overline{x_i}=x_i^{-1}$. Let $M$  be a $R$-module. The $R$-module $\overline{M}$ is obtained using the involution as follows. The elements of  $\overline{M}$ consist of the elements $m\in M$, with the action of $R$ defined by $(r,m)\mapsto \overline{r}m$.
\end{definition}

We note that the $R$-module given by the formal power series  $\sum_{\lambda \in \Z^n} c_\lambda x^\lambda$, with  $c_\lambda \in \Z/p$, is obtained by the map
\begin{equation}
    c_\lambda \in \bigoplus_{\lambda \in \Z^n} \Z/p  \mapsto \sum_{\lambda \in \Z^n} c_\lambda x^\lambda \in R
\end{equation}
and can be identified with finitely supported functions $\Z^n\to \Z/p$. Under this identification, we see that $P$ is a free $R$-module expressed as $P\cong R^{2m}$, and each element $p\in P$ corresponds to a $2m$-tuple $p=(p_1,p_2,\ldots,p_{2m})$.
Furthermore, $P$ can be equipped with an unimodular sesquilinear form $\Omega(\cdot, \cdot): P\times P\rightarrow R$, encoding the commutation relations of clock and shift operators. Equivalently, $\Omega$ can be seen as an isomorphism 
\begin{equation}
    P\rightarrow \Hom_R(\overline{P}, R): a\mapsto \Omega(\cdot, a).
\end{equation}
Usually $\Omega$ is taken to the hyperbolic form 
\begin{equation}
    \Omega(a, b):=\overline{a}^t\begin{pmatrix}
        0 & I_p\\
        -I_p & 0
    \end{pmatrix} b,
\end{equation}
for $a, b\in P$. We now define Pauli stabilizers in the language of modules for the ring $R$.

\begin{definition}
    A \emph{translation-invariant Pauli stabilizer code} is (represented by) an isotropic submodule $L\subset P$, i.e.
\begin{equation}
    L\subset L^\perp :=\{a\in P: \Omega(a, l)=0 \text{ for all }l\in L\}.
\end{equation}
We denote a translation-invariant Pauli stabilizer code by the pair $\fC = (L,P)$. For the remainder of this work we will stop explicitly stating ``translation-invariant''.
\end{definition}
We are mainly interested in studying the \textit{Lagrangian} stabilizer codes, which are submodules satisfying $L=L^\perp$. In this case the stabilizer code is said to be $\emph{topological}$ and frequently exhibit features believed to characterize gapped phases that are distinct from product states. 

\subsubsection{Topological Charges}\label{subsection:gs}
We recall the definition and properties of topological charges first proposed in~\cite{ruba2024homological}. A Lagrangian stabilizer code \(L \subset P\) determines a local Hamiltonian whose ground state is the unique pure state fixed by every element of \(L\). Equivalently, the ground state is stabilized by the elements of \(L\), the \emph{stabilizers}. Starting from \(L \subset P\), one can also construct excited states that differ from the ground state by point-like, string-like, or membrane-like excitations. These excitations are created by extended operators and correspond, at the algebraic level, to involution-antilinear homomorphisms \(L \to R\) that do not extend to \(P\). Their classification is identified with the module
\begin{equation}\label{eq:charge}
    \cE^i(\fC) := \Ext^{i+1}_R\!\left(\overline{P/L},R\right),
\end{equation}
where \(i=0,1,2,\dots\) corresponds respectively to point-like, string-like, membrane-like, and higher-dimensional excitations. We say that \(\cE^i(\fC)\) is mobile if it has Krull dimension zero.
In this work, we restrict attention to stabilizer codes that admit only ``mobile excitations''.

As shown in \cite[Section 6.3]{ruba2024homological} there are (translation invariant) braiding relationships between the operators $\cE^{i}(\fC)$ and $\cE^{j}(\fC)$ in $d$-spacetime dimensions given by a pairing $\mathfrak{P}$ with the property that $$\mathfrak{P}(\cE^{i}(\fC),\cE^{j}(\fC)) = -(-1)^{(i+1)(d-i-2)}\overline{\mathfrak{P}(\cE^{j}(\fC),\cE^{i}(\fC))}\,,$$ and is defined with respect to a cup product on cohomology and symplectic pairing $\omega: \overline{P/L}\otimes_R L \to R$. In analogy with TQFTs, in (2+1)d the braided reflects the statistics of line operators which is also captured by the $S$-matrix. In higher dimensions, a generalized version of the $S$-matrix for semisimple higher categories is work in progress \cite{JFSmatrix}. 

\subsubsection{General Stabilizer Code}
More generally, our starting point could be a non-hyperbolic form $\Omega$. A stabilizer code consists of a free $R$-module $P$ equipped with a unimodular sesquilinear form $\lambda: P \xrightarrow{\sim} \Hom_R(\overline{P}, R): a \mapsto \Omega(\cdot, a)$ and an $R$-submodule $L\subset P$ satisfying the Lagrangian property $L^\perp=L.$

Each stabilizer code gives a perfect chain complex making the topological charges explicit. Take a free resolution 
\begin{equation}
    \cdots \longrightarrow F_2 \longrightarrow F_1 \longrightarrow P \longrightarrow P/L \longrightarrow 0.
\end{equation}
As $R$ is regular, such a resolution may be chosen to have length bounded by the Krull dimension of $R$. 
We use the notation that $M^*=\mathrm{Hom}_R(\overline M, R).$ Consider the chain complex
\begin{equation}\label{eq:dualF}
F_\bullet:\qquad \cdots \longrightarrow F_2 \longrightarrow F_1 \xrightarrow{\sigma} P \xrightarrow{\delta} F_1^* \longrightarrow F_2^* \longrightarrow \cdots,
\end{equation}
where \(\delta\) is the composition
\[
P \xrightarrow{\lambda} P^* \xrightarrow{\sigma^*} F_1^*.
\]
This complex is exact at \(P\) and to the left of \(P\), while its cohomology groups to the right are given by
\[
\Ext_R^{i+1}\!\left(\overline{P/L},R\right).
\]
In particular, the Lagrangian condition \(L=L^\perp\) is equivalent to exactness at \(P\). Moreover, upon choosing a bounded resolution, the complex \(F_\bullet\) becomes perfect. We regard $F_\bullet$ as a homological presentation of the stabilizer code in that it packages both the stabilizer data and the associated excitation theory into a single perfect chain complex. Later, we will prove a form of \textit{Poincar\'e duality} for $F_\bullet$ arising from a mobile code. 

\begin{rem}
     We mention a general theory of complexes like $F_\bullet$ is being developed into a general theory \cite{shirley_wickenden_yang_beaudry_haah_inprogress, hermele2026fractonlandscape} of abelian gapped phases. Their theory goes beyond fully mobile models which is the main focus of this article. 
\end{rem}

\subsection{Main Results}\label{subsection:mainresults}
We classify Lagrangian stabilizer codes in arbitrary dimensions up to gapped interface and coarse-graining. By gapped interface we mean up to algebraic surgery of the complex $F_\bullet$, in the sense of $\rL$-theory for rings. We note that there is also a notion of gapped boundary and interface defined in terms of commuting, interpolating stabilizers~\cite{kitaev2012models, liang2024operator}. At least in spacetime dimension 3, the two notions agree; see Theorem~39 of~\cite{RY}. Coarse-graining is a specific operation that involves modding out by subgroups of the  translation symmetries.

\begin{notation}
    Anytime we refer to a theory being in $n$-dimensions, we take $n$ to mean the \textbf{spacetime} dimension, rather than only the spatial dimension. This is to match  with the mathematical notation. Furthermore a lattice $\Z^{n}$ will always denote a spacetime lattice of dimension $n$. This notation is motivated in part by its application to the classification of Clifford QCA. Rather than viewing a Clifford QCA as a purely spatial transformation, it is more natural to regard it as a discrete-time dynamical process with time direction $\Z$. We therefore record this structure by extending the lattice to spacetime.
\end{notation}


\begin{proposition}
Let $n>4$ and let $(L,P)$ be a $(n-1)$-dimensional mobile Lagrangian stabilizer code defined with respect to the ring $R =\Z/p[x_1^{\pm}, \dots, x_{n-2}^{\pm}]$. The collection of mobile excitations of $(L,P)$ is a Poincaré object in the stable $\infty$-category $(\cD^\perf(R),\Sigma^{n-2}Q)$.
\end{proposition}

\begin{theorem}\label{mainthm}
 Pauli stabilizer codes in $(n-1)$-dimensions where $n\geq 4$ are classified up to gapped interface and coarse-graining by the group $\rL_{-n}(\Z/p)$.
\end{theorem}
The precise meaning of gapped interface and coarse-graining will be given in \S\ref{subsection:targetstabilizer}, in particular Notation \ref{notation:gappedinterface}, and \S\ref{subsection:coarse}. The computation for $\rL$-groups are known and given as follows, for $p=2$ and $p\neq 2$:

\begin{equation}
    \rL_n(\Z/2) =  \begin{cases}
  0 & n \equiv 1,3 \mod 4 \\
  \Z/2  & n \equiv 2 \mod 4 \\
  \Z/2 & n \equiv 0 \mod 4\,
\end{cases}
\quad 
\rL_n(\Z/p) =  \begin{cases}
  0 & n\equiv 1,3 \mod 4 \\
  0  & n \equiv 2 \mod 4 \\
  \Witt(\Z/p) & n\equiv 0 \mod 4\,.
\end{cases}
\end{equation}
where $\Witt(\Z/p)$  is the Witt group of $\Z/p$, consisting of stable equivalences classes of nondegenerate quadratic forms over $\Z/p$.

In 4d, the authors \cite{Barkeshli:2022edm} give a gauge theories presentation of the theory that arises from the Pauli stabilizer codes.
The case of $3$-dimensional Pauli stabilizers was studied in \cite{RY}, where it was shown how to relate stabilizer states to a class of abelian anyon models that admit topological/gapped boundary
conditions. In general there can be more Pauli stabilizers which do not admit gapped boundary, and they represent nontrivial classes in the pointed Witt group. Finally, if $n< 3$ there are no interesting topological Pauli stabilizer codes.

By observing Theorem \ref{mainthm},  we observe a bulk-boundary relationship between the classification of Clifford QCAs and Pauli stabilizer codes. 

\begin{corollary}\label{cor:cliffordPS}
 The classification of Clifford QCAs in $n$-dimensions is equivalent to the classification of mobile Pauli stabilizer codes in $(n-1)$-dimensions. 
\end{corollary}
\noindent This result is analogous to the expected classification of TQFTs, where the bulk determines boundary theories up to equivalence via interfaces. Equivalently, while a given bulk may admit many boundary conditions, they are all related by gapped interfaces. A full  comparison to the classification of framed TQFTs will be the main point of discussion in \S\ref{section:framedTO}.

Our strategy for conducting the classification involves conjecturing the existence of a universal target category for topological stabilizers codes. Like with TQFTs, this target category has objects which are stabilizer codes in all dimension. 
The universal target should have the property that by only restricting to invertibles of the category, i.e. the core, we recover the spectrum of QCA.
Thus far, the spectrum of QCAs defined with respect to a ring $R$  has been defined in \cite{Ji:2026fka} and denoted $\mathbf{Q}$. The authors have shown that it enjoys the following property:

\begin{equation}
    \mathbf{Q}(\pt) \cong \Omega \mathbf{Q}(\Z),\, \mathbf{Q}(\Z) \cong \Omega \mathbf{Q}(\Z^2),\,\ldots, \mathbf Q(\Z^{n-1}) \cong \Omega \mathbf Q(\Z^n)\,.
\end{equation}
From now on we will use $\Q^n$ to denote $\Q(\Z^n)$ for unitary QCA spectrum defined in Section~5 of~\cite{Ji:2026fka}. Following this notation, we provide the following conjecture.
\begin{conjecture}\label{conj:universaltarget}
   There is a categorical spectrum of topological stabilizer codes $\TS^*$, such that $\TS^{n-1} \cong \Omega \TS^n$ for all $n\geq 1$, the core of $\TS^*$ is given by $\mathbf Q^*(\pt)$,  and for any other categorical spectrum $\mathbf{C}^*$ there is a natural in $\mathbf{C}$ isomorphism $\pi_0\hom_{\mathbf{CatSp}}(\mathbf{C}^*,\mathbf{TS}^*) = \hom_{\mathbf{ComAlg}}(\mathbf{C}^0,\mathbf{Q}^0(\pt))$. .
\end{conjecture}

Although the development of the category of topological stabilizers remains work in progress, clarifying its relationship with QCAs will render the classification of Pauli stabilizers precise. Approaching the problem in this manner also aligns the classification of topological lattice theories more closely with the continuum perspective, particularly with the program of classifying framed topological quantum field theories.

\subsection{Analogies Between Pauli Stabilizer Codes, TQFTs, and Manifolds}
One of the central goals of this work is to compare and contrast topological phases formulated on the lattice and the continuum. In carrying out this comparison, we find that several notions that arise in the two settings are in fact closely related. To make these correspondences precise, we summarize the key structural features on both sides and present their analogies in Table \ref{table:analogy}. We now elaborate on several of the analogies in the table, clarifying points that may initially appear opaque:\\

\begin{itemize}
\item (2): A perfect chain complex encodes a strong finiteness condition: it is, up to quasi-isomorphism, a bounded complex of finitely generated projective modules. In a similar spirit, full dualizability imposes a stringent finiteness condition on categorical objects. This parallel helps explain why TQFTs are naturally formulated in terms of fusion categories.
\item (3): A perfect pairing endows a chain complex with the structure of a Poincaré complex, i.e. a chain complex that is quasi-isomorphic to its dual, expressing a nondegeneracy condition. Analogously, in TQFT, triviality of the center imposes a nondegeneracy condition, such as the invertibility of the $S$-matrix.
\item (5): Spin–statistics relates the spin of an excitation to its exchange behavior, determining whether it obeys bosonic or fermionic statistics. On the geometric side, a quadratic refinement of the mod 2 intersection pairing is closely tied to the choice of a spin structure, which is required to consistently define fermions on a manifold.
    \item (6): Translation symmetry implies that one can move the system along directions of a compact space and return to the same point without altering the physics. In a similar vein, the fundamental group parametrizes the distinct ways of transporting around loops in the manifold.
    \item (7): Coarse-graining in a stabilizer code can reduce translation symmetry: for example, translation by one lattice unit may cease to be a symmetry, while translation by two units persists. Geometrically, this is analogous to passing to a double cover of a manifold, where the finer translation symmetry is restored at the level of the cover.
\end{itemize}

\begin{table}[ht]
\centering
\renewcommand{\arraystretch}{1.3}
\setlength{\tabcolsep}{8pt}
\begin{tabularx}{\textwidth}{%
    >{\centering\arraybackslash}c
    >{\centering\arraybackslash}X
    >{\centering\arraybackslash}X
    >{\centering\arraybackslash}X}
\toprule
 & \textbf{Pauli Stabilizer Codes} & \textbf{TQFTs} & \textbf{Manifolds} \\
\midrule
1 & Mobile Excitations & $n$-Morphisms & Homology Modules \\
2 & {\small Finite Superselection Sectors} & Full Dualizability & Finite Dimensionality \\
3 & Perfect Braiding & Trivial Center & Poincaré Duality \\
4 & Algebraic Surgery & Morita Equivalence & Surgery \\
5 & Spin-Statistics & Spin-Statistics & Quadratic Refinement \\
6 & Translation Symmetry & \text{---} & Fundamental Group \\
7 & Coarse-Graining & \text{---} & Lifting to a Finite Cover \\
\bottomrule
\end{tabularx}
\vspace{3mm}
\caption{Table of analogies between concepts that appear in the theory of Pauli stabilizer codes, topological field theory, and manifolds.}
\label{table:analogy}
\end{table}

\subsection{Outline}
The paper is organized as follows. In \S\ref{subsection:cliffordQCA}, we review the classification of Clifford QCAs and their relationship to Pauli stabilizer codes, as this will play a key role in the classification results that follow. In \S\ref{subsection:targetstabilizer} we  define the universal target category for topological stabilizer codes as a categorical spectrum and describe its properties.  
In \S\ref{section:backgroundLtheory} we present a modern treatment of algebraic $\rL$-theory, covering all the background that will be needed to define algebraic surgery. In \S\ref{subsection:coarse} we describe the coarse-graining procedure for stabilizer codes. These ingredients are essential, as they supply the equivalence relations underlying the classification. The main results are then formulated and proved in \S\ref{subsection:surgery}. In \S\ref{section:framedTO} we turn to the classification of TQFTs, and explore the similarities that it shares with the classification of Pauli stabilizer codes. Crucially, \S\ref{subsection:latticeandcontinuum} explains concrete  differences between topological lattice phases and topological phases in the continuum.

\section{QCA, Stabilizer Codes, and a Universal Target}\label{section:universaltarget}

\subsection{Clifford QCAs and Pauli Stabilizer Codes}\label{subsection:cliffordQCA}
We begin with a review of Clifford QCAs and their classification, following \cite{Yang:2025jvn}.  This information will come into use later when we compare with the classification of Pauli stabilizers.
We then review the space of (generic) QCAs on a point.


We start by recalling the algebraic structure of a quantum spin system. Let $\Lambda$ be a set equipped with a metric $\rho$, which we refer to as a lattice.\footnote{As is standard in condensed matter physics, the term ``lattice'' here refers only to a discrete metric space, and not necessarily to a group.} For each finite subset $X\subset \Lambda$, let
\[
\cA_X := \bigotimes_{x\in X}\Mat(\bC^{d_x}),
\]
where $\Mat(\bC^{d_x})$ is the full matrix algebra on the local Hilbert space at the site $x$. If $Y\subset X$, we define the inclusion
\[
\iota_{Y,X}:\cA_Y\to \cA_X,\qquad \iota_{Y,X}(A)=A\otimes \id_{X\setminus Y}.
\]
In this way, the collection $\{\cA_X\}_{X\in \mathcal P_0(\Lambda)}$ forms a directed system indexed by the finite subsets of $\Lambda$.

An observable $A\in \cA_X$ is said to be supported on $Y\subset X$ if it may be written in the form
\[
A=A'\otimes \id_{X\setminus Y}
\]
for some $A'\in \cA_Y$. Its support is then the minimal such subset. The algebra of local observables is defined to be the direct limit
\[
\cA_{\loc}:=\varinjlim_{X\in \mathcal P_0(\Lambda)} \cA_X.
\]
\begin{proposition}
The algebra $\cA_{\loc}$ is a $*$-algebra equipped with the operator norm, and its norm completion
\[
\cA:=\overline{\cA_{\loc}}
\]
is a $C^*$-algebra.
\end{proposition}

\begin{proof}
See Section 6.2 of \cite{bratteli2012operator}.
\end{proof}

\begin{definition}
A \emph{quantum cellular automaton} (QCA) is a $*$-automorphism $\alpha$ of $\cA$ for which there exists a constant $r\geq 0$, called the \emph{range} of $\alpha$, such that for every $i\in \Lambda$ and every $A\in \cA_{\{i\}}$ one has
\[
\alpha(A)\in \cA_{B_r(i)},
\]
where
\[
B_r(i):=\{j\in \Lambda:\rho(i,j)\leq r\}.
\]
\end{definition}

Thus, a QCA is an automorphism that spreads a strictly local observable only a bounded distance. In this sense, QCAs may be viewed as discrete-time locality-preserving evolutions. The set $\Q(\cA)$ of all QCAs forms a group under composition.

\begin{example}
Let $\Lambda=\Z$, and suppose the single-site algebras are identified by isomorphisms
\[
t_i:\cA_{\{i\}}\to \cA_{\{i+1\}}.
\]
Then the right shift
\[
\tau(A)=t_i(A),\qquad A\in \cA_{\{i\}},
\]
defines a QCA.
\end{example}

\begin{example}
More generally, if $T:\Lambda\to \Lambda$ is a bijection such that $\rho(T(i),i)$ is uniformly bounded, then permuting the single-site algebras according to $T$ defines a QCA.
\end{example}

A second fundamental source of examples comes from quantum circuits. These furnish the basic notion of trivial or finite-depth locality-preserving evolution.

\begin{definition}
A \emph{quantum gate} is a local unitary
\[
G\in \mathrm U(\cA_X):=\{A\in \cA_X:A^*A=AA^*=I\}
\]
supported on some finite subset $X\subset \Lambda$. A \emph{single-layer quantum circuit} is a formal product
\[
\prod_{i\in I} G_i,
\]
where the supports of the gates $G_i$ are pairwise disjoint and uniformly bounded. Such a layer defines a QCA by conjugation. A \emph{finite-depth quantum circuit} is a finite composition of single-layer circuits.
\end{definition}

Finite-depth circuits form the basic equivalence relation in the theory of QCAs.

\begin{proposition}
Conjugation by finite-depth quantum circuits forms a normal subgroup
\[
\C(\cA)\triangleleft \Q(\cA).
\]
\end{proposition}

\begin{proof}
See Lemma 2.10 of \cite{freedman2020classification}.
\end{proof}

There is also a stabilization procedure, analogous in spirit to stabilization in algebraic $\rK$-theory, obtained by adjoining ancillary degrees of freedom. If $\cA$ and $\cA'$ are local observable algebras over the same metric space $\Lambda$, then their tensor product is defined sitewise, and every QCA $\alpha$ on $\cA$ induces a QCA $\alpha\otimes \id_{\cA'}$ on $\cA\otimes \cA'$. Iterating this construction yields an ascending sequence
\[
\Q(\cA)\hookrightarrow \Q(\cA^{\otimes 2})\hookrightarrow \cdots,
\]
and we write $\Q(\Lambda)$ for the resulting union. Similarly one defines $\C(\Lambda)$.

The relevance of stabilization is that it exposes the large-scale algebraic structure of QCAs after finite-depth trivialities have been quotiented out.

\begin{theorem}
The subgroup $\C(\Lambda)$ is normal in $\Q(\Lambda)$, and the quotient
\[
\overline{\Q}(\Lambda):=\Q(\Lambda)/\C(\Lambda)
\]
is abelian.
\end{theorem}

\begin{proof}
See Theorem 2.3 of \cite{freedman2020classification}.
\end{proof}

We now restrict to the Pauli setting, which is the class relevant for Clifford QCAs and stabilizer codes. Let $\mathcal H=\bC^d$ be a $d$-dimensional Hilbert space, with distinguished basis $\{|0\rangle,\dots,|d-1\rangle\}$. The generalized Pauli operators, defined in \S\ref{subsection:introtopauli}, act on these states as
\[
X|k\rangle=|k+1 \!\!\!\!\mod d\rangle,
\qquad
Z|k\rangle=\xi^k |k\rangle,
\qquad
\xi=e^{2\pi i/d}.
\]


Let $\Lambda$ be a metric space, where at each site $i\in \Lambda$ we assign a tensor power $\mathcal H^{\otimes n_i}$. Let $\cA$ be the norm completion of the corresponding local observable algebra, and denote by $X_{k,i}$ and $Z_{k,i}$ the generalized Pauli operators acting on the $k$th qudit at site $i$. The \emph{generalized Pauli group} $\mathcal P(\cA)$ is the subgroup of $\mathrm U(\cA)$ generated by all $X_{k,i}$, $Z_{k,i}$, and the scalar phases $\lambda I$ for $\lambda\in \bC^\times$.
For notational simplicity, we temporarily suppress the qudit index and assume there is a single qudit at each site. Since the Pauli operators generate the local matrix algebra, a QCA is determined entirely by its action on this subgroup.

\begin{proposition}
A QCA $\alpha$ is uniquely determined by the images $\alpha(X_i)$ and $\alpha(Z_i)$ for all $i\in \Lambda$. Equivalently, it suffices to know the induced injective group homomorphism
\[
\bar\alpha:\mathcal P(\cA)\to \mathrm U(\cA)=\{A\in \cA: A^*A=AA^*=I\}.
\]
\end{proposition}

\begin{proof}
Since $\mathcal P(\cA)$ generates $\cA$ as an algebra and $\alpha$ is an automorphism, the claim follows.
\end{proof}

This reduction to the Pauli group is the key reason that Clifford QCAs admit an algebraic description in terms of symplectic modules.

\begin{definition}
A QCA $\alpha$ is called \emph{separated} if
\[
\alpha(\langle X_i\rangle_{i\in \Lambda})=\langle X_i\rangle_{i\in \Lambda}
\qquad\text{and}\qquad
\alpha(\langle Z_i\rangle_{i\in \Lambda})=\langle Z_i\rangle_{i\in \Lambda}.
\]
\end{definition}

Translations are immediate examples of separated QCA, since they merely permute the $X_i$ and $Z_i$ operators independently.

\begin{definition}\label{defn:CQ}
A \emph{Clifford QCA} is an automorphism $\alpha$ of the Pauli group $\mathcal P(\cA)$ with bounded propagation, in the sense that there exists $r>0$ such that for every $i\in \Lambda$, the operators $\alpha(X_i)$ and $\alpha(Z_i)$ are products of Pauli operators supported within distance $r$ of $i$. We denote the group of Clifford QCAs by $\cQ(\cA)$.
\end{definition}

\begin{definition}
A quantum gate $G$ is \emph{Clifford} if
\[
G\mathcal P(\cA)G^{-1}=\mathcal P(\cA).
\]
A \emph{Clifford circuit} is a finite-depth circuit consisting of Clifford gates. We denote by $\cC(\cA)$ the subgroup of QCAs given by conjugation by Clifford circuits.
\end{definition}

\begin{example}
The product
\[
\mathcal X=\prod_{i\in \Lambda} X_i
\]
defines a Clifford circuit. Indeed,
\[
\mathrm{Ad}_{\mathcal X}(X_i)=X_i,
\qquad
\mathrm{Ad}_{\mathcal X}(Z_i)=\xi Z_i.
\]
\end{example}

To pass to the algebraic description, one quotients the Pauli group by its center. The resulting group
\[
P(\Lambda,d):=\mathcal P(\cA)/Z(\mathcal P(\cA))
\]
is abelian, and admits the convenient description
\[
P(\Lambda,d)=\bigoplus_{i\in \Lambda} P(i)
\cong \bigoplus_{i\in \Lambda} (\Z/d)^2.
\]
At each site, the operators $X_i$ and $Z_i$ correspond respectively to the basis vectors
\[
\begin{pmatrix}1\\0\end{pmatrix},
\qquad
\begin{pmatrix}0\\1\end{pmatrix}.
\]
The noncommutativity of the Pauli group is encoded, after passing to the quotient, by the standard symplectic form
\[
\omega:\Z_d^2\times \Z_d^2\to \Z_d,
\qquad
\omega\!\left(
\begin{pmatrix}a\\b\end{pmatrix},
\begin{pmatrix}e\\f\end{pmatrix}
\right)=af-be.
\]
Summing over all sites yields a global symplectic form
\[
\Omega:=\bigoplus_{i\in \Lambda}\omega
\]
on $P(\Lambda,d)$. We write simply $P$ for this symplectic abelian group.

The point of this construction is that Clifford QCAs descend to locality-preserving symplectic automorphisms of $P$. In other words, once one quotients by phases, the operator-algebraic problem becomes a linear symplectic one.

\begin{definition}\label{defn:localsymp}
A \emph{local symplectic automorphism} of $P$ consists of homomorphisms
\[
\alpha_j^i:P(i)\to P(j), \qquad i,j\in \Lambda,
\]
such that:
\begin{enumerate}
    \item the induced map preserves the symplectic form $\Omega$;
    \item it is invertible;
    \item there exists a constant $r>0$ such that $\alpha_j^i=0$ whenever $\rho(i,j)>r$.
\end{enumerate}
We denote the group of such automorphisms by $\Aut_{\loc}^{\Omega}(P)$.
\end{definition}

The relation between Clifford QCAs and local symplectic automorphisms is the basic algebraic mechanism underlying the classification theory.

\begin{lemma}\label{lem:kappa}
There is a surjective homomorphism
\[
\kappa:\cQ(\cA)\to \Aut_{\loc}^{\Omega}(P).
\]
Moreover,
\[
\ker \kappa \subset \cC(\cA).
\]
\end{lemma}

\begin{proof}
A Clifford QCA preserves the center of $\mathcal P(\cA)$ and therefore descends to an automorphism of
\[
P=\mathcal P(\cA)/Z(\mathcal P(\cA)).
\]
Since commutators in the Pauli group record the symplectic pairing, the induced automorphism preserves $\Omega$. The locality condition is inherited directly from bounded propagation.

If $\alpha\in \ker\kappa$, then $\alpha$ acts on each generator by a phase:
\[
\alpha(X_i)=\xi^{m_i}X_i,
\qquad
\alpha(Z_i)=\xi^{n_i}Z_i.
\]
Such phases may be implemented by a single-layer Clifford circuit, namely by conjugation with suitable products of $X_i^{n_i}Z_i^{-m_i}$. Hence $\ker \kappa \subset \cC(\cA)$.

Conversely, given a local symplectic automorphism of $P$, one chooses lifts of the images of the basis vectors corresponding to $X_i$ and $Z_i$ to Pauli operators in $\mathcal P(\cA)$, and then extends multiplicatively. The symplectic condition guarantees compatibility with the commutation relations, while locality guarantees bounded propagation. This produces a Clifford QCA lifting the given symplectic automorphism.
\end{proof}

\begin{definition}
    Let $K(\Lambda,\Z/p)$ denote the group of stabilized Clifford QCAs up to Clifford circuits and separated QCAs on $\Lambda$.
\end{definition}

In the case where $\Lambda =\Z^n$, the following equivalence  can be established following \cite[Corollary 47]{Yang:2025jvn}:
\begin{theorem}\label{thm:ClassifyCliffordQCA}
    There is a group isomorphism $K(\Z^n,\Z/p) \cong \rL_{-n}(\Z/p)$.  
\end{theorem}
This implies the Clifford QCAs in dimension $n$ are classified by the following groups:
\begin{equation}\label{eq:valuesofK}
   K(\Z^n,\Z/p)= \begin{cases}
  0 & n \equiv 1,3 \mod 4 \\
  0  & n \equiv 2 \mod 4 \\
  \Witt(\Z/p) & n \equiv 0 \mod 4\,.
\end{cases}
\end{equation}
We now proceed to describe the space of all QCAs, and review how the space forms an $\Omega$-spectrum.

\begin{definition}
    An \textit{$\Omega$-spectrum}\footnote{There are many different yet equivalent ways to define spectra; see for example~\cite{MMSS01}. We use $\Omega$-spectra because they tend to appear in physics applications: see, for example, \cite{Kit13, Freed:2014iua, Kit15, GJF19}.
    }
$E$ is a sequence of pointed topological spaces $\{E_n\}_{n \in \mathbb{Z}}$ together with structure maps, which are homotopy equivalences:
\begin{equation}
\sigma_n : E_n \overset\simeq\to \Omega E_{n+1}
\end{equation}
where $\Omega$ denotes the based loop space functor.
\end{definition}
In the category of spectra, the inverse of taking loops is suspension.  The suspension of a spectrum $\Sigma E$  is defined by $ \Sigma E_{n} \cong  E_{n+1}$.
Having introduced QCAs as locality preserving automorphisms on the algebra of observables of quantum spin systems, a natural way to think about them is to realize them as morphisms in a category of quantum spin systems on a lattice $\Z^n$.

\begin{definition}
    The category of quantum spin systems on a metric space $X$ is the category $\mathbf S(X)$ described by the following globular space:
    \begin{itemize}
        \item the objects are quantum spin systems on $X$,
        \item the morphisms are locality preserving algebra homomorphisms.
    \end{itemize}
\end{definition}

\begin{definition}
  The \textit{space of QCAs over a point} $\Q(\pt)$ is the $\mathrm{K}$-theory space of the category of quantum spin systems $\mathbf{S}$. 
\end{definition}

\begin{theorem}[\cite{Ji:2026fka}]
The space $\Q(\pt)$ has the following looping property:
     $$\mathbf{Q}(\pt) \cong \Omega \mathbf{Q}(\Z),\quad  \mathbf{Q}(\Z) \cong \Omega \mathbf{Q}(\Z^2),\quad \ldots \quad  \mathbf Q(\Z^{n-1}) \cong \Omega \mathbf Q(\Z^n)\,.$$
\end{theorem}

A separate construction for a spectrum of QCAs involving blend equivalences and globular sets is given in \cite{Czajka:2025mme}. A direct consequence of the proceding theorem gives control over maps between invertible stabilizers:
\begin{lemma}\label{lemma:invertiblestabilizer}
    A map between the trivial $n$-dimensional invertible stabilizer $\mathbf 1$ and a nontrivial $n$-dimensional invertible stabilizer $\alpha$ is implemented by an invertible stabilizer in $(n-1)$-dimensions if and only if $\alpha \cong  \mathbf 1$.
\end{lemma}

We will argue for this result physically. The loop space in the space of $n$-dimensional QCAs is a family of $n$-dimensional QCAs that start and end at the identity.  At macroscopic length
scales, this continuous transition from the identity to identity can be squeezed into an $(n-1)$-dimensional system if it is a ``gapped'' transition i.e. one which does not depend on the geometry of the path. This gives a map $\Omega \Q(\Z^n) \to \Q(\Z^{n-1})$. To give a map in the other direction, we take a $(n-1)$-dimensional QCA $\alpha$ and consider it as an invertible defect in the trivial $n$-dimensional QCA. We then populate the phase with insertions of $\alpha$  and its inverse $\alpha^{-1}$, separated at a mesoscopic scale, repeating the insertions across the entire new dimension.  This system created by interleaving $\alpha/\alpha^{-1}$ is connected to the trivial QCA by a continuous deformation in two different ways, and hence gives a loop in $\Q(\Z^n)$ from the trivial QCA to itself. One finally needs to show that it is possible to provide nullhomotopies for the composition $\Omega \Q(\Z^n) \to \Q(\Z^{n-1}) \to \Omega \Q(\Z^n)$ and $\Q(\Z^{n-1}) \to \Omega \Q(\Z^{n}) \to  \Q(\Z^{n-1})$. For an argument of these facts for continuum field theories, see \cite[Section 3.1]{GJF19}.  

\begin{rem}
    There is an analogy of $\Q(\pt)$ with $I_{\mathbb{C}^\times}(\pt)$, the latter which classifies invertible framed topological field theories. The group $\pi_{-n}(I_{\mathbb{C}^\times}(\pt))$ are the isomorphism classes of $n$-dimensional invertible framed TQFTs, while the group $\pi_{-n}(\Q(\pt))$ are isomorphism classes of $n$-dimensional QCAs. Moreover, there is an analog for invertible field theories in the lattice setting, and these are invertible stabilizer codes.
    Such codes can be realized as the ground states of a class of Hamiltonians with finite-range interactions. A Clifford QCA acts on $P$ via a symplectic automorphism, and if  there exists a submodule $K\subset P$ consisting only of products of $Z$ operators such that for $\alpha$ a Clifford QCA, $\alpha(K) = L$, then the stabilizer $(L,P)$ is created by $\alpha$ \cite[Definition 60]{RY}. It was conjectured in \cite{Shuklin:2025mhu} that the equivalence classes of invertible stabilizer codes coincide with those of quantum cellular automata (QCA) modulo finite-depth quantum circuits and lattice translations. 
\end{rem}

\subsection{A Universal Target for Topological Stabilizer Codes}\label{subsection:targetstabilizer}

Our inspiration for performing the classification of Pauli stabilizers is due to  Johnson-Freyd-Reutter, in their work on developing higher algebraic closure and the universal target category for semisimple $n$-categories \cite{JFtalk1,JFtalk2,JFtalk3,JFtalk4,JFSDVS,DavidTalk}. Integral in their understanding of the universal target category is the theory of TQFTs, and the spectrum $I_{\mathbb{C}^\times}(\pt)$.

 Given that there is a spectrum for QCAs, we will define the category of stabilizer codes. We do so using the concept of a categorical spectrum.\footnote{The name categorical spectrum is a concept that has many origins. See \cite{masuda2024} for recent developments in the subject.} 

\begin{definition}
    A categorical spectrum $\cT$ consists of:
    \begin{itemize}
        \item for each $n \geq 0$ a symmetric monoidal $(\infty, n)$-category $\cT^n$ with unit object $\mathds 1_n$, and 
        \item for each $n \geq 1$ a specified equivalence of $(\infty,n-1)$-categories $\End_{\cT^n}(\mathds 1_n) \cong \cT^{n-1}$\,.
    \end{itemize}
\end{definition}
\noindent The invertible part of a categorical spectrum is extracted by taking Picard $\infty$-groupoids levelwise. For an $(\infty,n)$-category $\cC$ let $(\cC)^\times$ be the maximal $\infty$-subgroupoid of invertible objects and invertible $k$-morphisms for $k\leq n$. The sequence of spaces $(\cT^n)^\times$ assemble into a spectrum $(\cT)^\times$.


In many examples, $\cT^0$ is a ``classical'' object i.e. one of categorical level 0 or 1, while $\cT^k$ is obtained by iterative categorification. Hence the new interesting information grows to the cohomological
direction i.e. the negative homotopy groups. In contrast, spectra have new information in the homological direction which captures the derived information of higher algebra. For this reason Johnson-Freyd has advocated using the terminology of \textit{deeper algebra} to describe categorical spectra \cite{JFtalk,JFtalk2025}.\footnote{The category of categorical spectra has been given in \cite{stefanich:thesis}. It was also credited to Claudia Scheimbauer in \cite{JFtalk}.}

\begin{definition}
    The \textit{universal target category} for topological stabilizer codes is the categorical spectrum $\TS^*$, defined by the property that  $(\TS^*)^\times\cong \Q^*(\pt)$ and for any other categorical spectrum $\mathbf{C}$ there is a natural in $\mathbf{C}^*$ isomorphism $\pi_0\hom_{\mathbf{CatSp}}(\mathbf{C}^*,\mathbf{TS}^*) = \hom_{\mathbf{ComAlg}}(\mathbf{C}^0,\mathbf{Q}^0(\pt))$. 
\end{definition}
\begin{rem}
     The category $\TS$ is defined for general topological stabilizer codes, though in this work we will only focus on the categorical spectrum of Pauli stabilizer codes, whose core is given by the spectrum of Clifford QCAs. While the problem of computing the low degree homotopy groups of $\Q(\Z^n)$ for $n>1$ i.e. the negative homotopy groups of $\Q(\pt)$ is difficult, we do know what the spectrum of 
      Clifford QCAs looks like: it is given by the $\rL$-theory groups $\rL_n(\Z/p)$. Therefore we have a map $\rL_n(\Z/p) \to \Q(\pt)$ which maps Clifford QCAs to all QCAs. In future work we aim to further understand the homotopical properties of $\mathbf{S}(\Z^n)$ and compute its algebraic $K$-theory groups. With this knowledge in hand, one can hope to better understand the finer details of $\TS$. 
\end{rem}

\begin{rem}\label{rem:semisimpletarget}
   The universal target category for framed TQFTs is denoted $\mathbf{W}$, see the most recent talk \cite{JFtalk4}. Every layer $\mathbf{W}^\bullet$ is a rigid symmetric monoidal $\mathbb C$-linear $n$-category with duals. The defining property of $\Mod(\mathbf{W}^{n-1})$ is that every object is connected to the unit object by a nonzero morphism. Since the category at each level is semisimple, the category $\Mod(\cC)$ is implemented by delooping $\cC$ and then condensation completing. We will discuss this in more detail in \S\ref{subsection:prelim}.  One can view $\mathbf W$ as a Morita category for semisimple $n$-categories, indeed the case $\mathbf W^2 = \mathbf{2sVect}$ which is the Morita category of super algebras. While we will not do it in this paper, we expect $\TS$ to be constructed in such a way that it also behaves like a Morita category.
\end{rem}


Using the target category for topological stabilizer codes $\TS$, we can now describe how to classify topological Pauli stabilizer codes up to gapped interfaces, and coarse-graining. In the remainder of this section, we make precise the notion of classifying stabilizers up to gapped interfaces. As will become clear, this equivalence may be understood as classification modulo suitable surgeries on the mobile operators.

Before arriving at the definition of gapped interfaces, we start with another very important concept which is that of a gapped boundary. For an invertible stabilizer code $\mathcal{S}$ in $n$-dimensions, one can ask about its possible \textit{gapped boundaries} aka modules. In particular a gapped boundary for $\cS$ is a topological interface between $\cS$ and the trivial theory. From an external perspective, one can study which of the objects in the family $(\TS^n)^\times$ admit a gapped boundary, on which resides a not necessarily invertible stabilizer code in $\TS^{n-1}$. Following Lemma \ref{lemma:invertiblestabilizer}, the theory on the gapped boundary is not invertible because we are not considering  a map from the trivial stabilizer to itself, but rather to a nontrivial invertible stabilizer.
 We denote this category $\Mod(\TS^{n-1})^\times$ whose objects are categories $\cS$ with a map $X:\mathds{1}_{\TS^n}\to \cS$, where $\mathds{1}_{\TS^n}$ is the trivial $n$-dimensional stabilizer. We summarize this discussion by the following result.

\begin{lemma}\label{lem:boundaryinv}
    There is an injective map on objects from $\Mod( \TS^{n-1})^\times \to (\TS^n)^\times$.
\end{lemma}



Given a bulk non-invertible stabilizer code $\cV$ and a  boundary theory $X$, the natural objects to study are the topological defects of different codimension within $X$. There is a natural relationship between the mobile operators in $\cV$ and the topological operators in $X$ given by a bulk-boundary relationship. 

\begin{example}
    For concreteness, we spell out this relationship in the case where $\cV$ is a (2+1)d theory. A {boundary} for $\cV$ is determined by an object in the category
\begin{equation}
    \text{bdry}(\cV)\simeq \Hom_{\TS^3}(\cV,\mathds 1_{\TS^3})\simeq \Mod_\cV.
\end{equation}
The collection of surface operators in $\cV$ is given by 
\begin{equation}
\text{surf}(\cV)\simeq \End_{\TS^3}(\cV). 
\end{equation}
The collection of line operators in $\cV$ is given by 
\begin{equation}\label{eq:codim2}
\text{line}(\cV)\simeq \End_{\End_{\TS^3}(\cV)}(\cV).
\end{equation}
Let $X$ be a (1+1)d theory at the boundary of $\cV$ i.e. $X \in \Mod_\cV$. Line operators in $X$ are endomorphisms of this boundary, and so correspond objects of $\End_{\Mod_\cV}(X)$. The bulk-boundary correspondence states that modules i.e. endings/boundaries for surface operators of $\cV$ give the line operators in the category $\End_{\Mod_\cV}(X)$.
\end{example}

Having seen that the objects of $\Mod(\TS^{n-1})$ are the $n$-dimensional topological stabilizer codes that admit a boundary condition, we can go further and incorporate boundary interfaces. Let $\Mod^2(\TS^{n-2})$ be the category whose objects are $n$-dimensional topological stabilizer codes that admit a boundary, up to boundary interface.\footnote{The exponent `2' denotes applying $\Mod$ twice.}

\begin{proposition}\label{prop:boundaryinterface}
    There is an essentially surjective map $\Mod^2(\TS^{n-2})\to \Mod(\TS^{n-1})$.
\end{proposition}
\begin{proof}
    An element  in the target is represented by $(\cV,X)$, where $\cV$ is an $n$-dimensional non-invertible stabilizer code and $X$ is a module for $\cV$.
    The category $\End_{\Mod_\cV}(X)$ determines the topological operators of all  codimension allowed on $X$, as well as the mobile excitations in the bulk theory $\cV$ by the bulk-boundary relationship. The data $(\cV,X)$ lifts to $\Mod^2(\TS^{n-2})$ precisely when $\End_{\Mod_{\cV}}(X)\in \TS^n$ lifts to $\Mod(\TS^{n-1})$. Thus it suffices to show that all the endomorphisms arise from modules of theories in 
    $\TS^{n-1}$.  We go by way of showing that given $F: X \to X$ we can assign a bimodule $M_F:X^{\mathrm{op}} \times X \to \mathbf{Set}$, such that this assignment preserves the identity and composition. For $x,y\in X$ we take 
    \begin{equation}
        M_F(x,y) = \hom(y,F(x)).
    \end{equation}
    For every morphism $a:y'\to y$ and $b:x\to x'$, we take $M_F(a,b): M_F(y,x) \to M_F(y',x')$, such that for $f\in \hom(y,F(x))$,
    \begin{equation}
        M_F(a,b)(f) = F(b)\circ f \circ a\,.
    \end{equation}
    One can immediately check that composition holds: given $a_1:y''\to y'$, $a_2:y'\to y$ and $b: x \to x'$, $b: x' \to x''$, then
    \begin{equation}
    M_F(a_1,b_2)\circ M_F(a_2,b_1) = M_F(a_2 \circ a_1,b_2 \circ b_1),
    \end{equation}
    using the functorality of $F$.

    We now show that composition of endofunctors is represented by composition of bimodules. If $P,Q$ are two bimodules, their composition $P\circ Q$ has objects defined as follows:
    \begin{equation}
        (P\circ Q)(y,x) = \left( \bigsqcup_{z\in X} P(y,z) \times Q(z,x) \right)/\sim 
    \end{equation}
    Where the relation $\sim$ is given by
    \begin{equation}
        (P(\id,r)(p),q) = (p,Q(r,\id)(q)), \quad \forall r:z \to z'\,.
    \end{equation}

    Let $G: X\to X$ be another endomorphism and consider $M_F(y,z) \circ M_G(z,x)$. An object in $M_F \circ M_G(y,x)$ is given by the class of a pair $(f:y\to F(z),g:z\to G(x))$. We define a comparison map 
    \begin{align}
        \Phi: (M_F \circ M_G)(y,x)&\to M_{F\circ G}(y,x)\\ \notag 
        (f,g)&\mapsto F(g)\circ f
    \end{align}
    and a map in the opposite  direction 
\begin{align}
    \Theta: M_{F\circ G}(y,x) &\to (M_F \circ M_G)(y,x) \\ \notag 
    h &\mapsto (h: y\to F(G(x)), \id: G(x)\to G(x))\,.
\end{align}
Showing that $\Phi \circ \Theta = \id$ is straightforward. We show that $\Theta \circ \Phi =\id$ as follows:
\begin{equation}
    (\Theta \circ \Phi)(f,g) = (F(g)\circ f,\id)\,,
\end{equation}
but we can use the equivalence relation, with $r=g : z\to G(x)$ to show that 
\begin{equation}
    (F(g)\circ f,\id)\sim (f,\id\circ g) = (f,g)\,.
\end{equation}
Thus $\Theta \circ \Phi =\id$ and we can conclude.
\end{proof}

In the study of boundary conditions, we have seen that a non-invertible boundary $X$ is equivalent to an interface between the trivial invertible stabilizer and some nontrivial invertible stabilizer $\cS$. A crucial  question we now need to understand is if there is a way to simplify $X$ via a sequence of topological operations such that it is invertible.
\begin{notation}\label{notation:gappedinterface}
   A topological operation on $X$ is implemented by  \textit{gapped interfaces} along $X$ if it is given by algebraic surgery applied to the mobile operators of $X$. If there is a sequence of surgeries that relate $X$ to an invertible theory, then there exists a trivialization of $\cS$.
\end{notation}
 This definition is motivated by the classification of TQFTs, owing to the close structural correspondence between mobile excitations in stabilizer codes and the topological excitations characteristic of TQFTs. We will elaborate on surgery from the side of TQFTs in \S\ref{subsection:classifyingTO}. In light of these definitions, we see that the objects in $\Mod^2(\TS^{n-2})^\times$ are invertible  $n$-dimensional stabilizers with a choice of gapped boundary condition, modulo interfaces between boundary conditions. Thanks to Lemma \ref{lem:boundaryinv} and Proposition \ref{prop:boundaryinterface} we obtain a map 
$\Mod^2(\TS^{n-2})^\times \to (\TS^n)^\times$.
The image of the map selects which $n$-dimensional invertible stabilizers admit genuinely nontrivial boundary conditions i.e. ones that are not related to an invertible boundary condition by gapped interface.
In \S\ref{section:Ltheory}, we will classify Pauli stabilizer codes in $(n-1)$-dimensions, which is equivalent to classifying the objects of $\Mod^2(\TS^{n-2})^\times$.

\section{Preliminaries on Algebraic $\rL$-theory}\label{section:backgroundLtheory}

Given that the purpose of this paper is to classify Pauli stabilizers in arbitrary dimension up to gapped interface, the natural mathematical framework in which to formulate the problem is algebraic 
$\rL$-theory. Stabilizer codes are not defined with  manifolds in mind but feel more algebraic in nature. In this regard, we require the framework of $\rL$-theory for rings and categories to classify them, and this is encapsulated in algebraic $\rL$-theory. We will take a modern derived perspective on algebraic $\rL$-theory following the notes \cite{lurie287x}, applied to the category of perfect chain complexes over $R$.
In particular, algebraic $\rL$-theory provides the appropriate receptacle for the obstruction classes that control whether a given Pauli stabilizer code can be reduced, via a sequence of surgical modifications implemented by gapped interfaces, to an invertible one. In this sense,  $\rL$-theory plays for stabilizer phases the same structural role that it plays in the surgery classification of manifolds: it detects precisely the residual quadratic data that cannot be eliminated.

 Like how $\mathrm{K}$-theory is used to extract algebraic invariants of stable $\infty$-categories, $\rL$-theory and $\rL$-groups can extract invariants of stable $\infty$-categories $\cC$ with quadratic functor $Q: \cC^{\op}\to \Sp$ to the category of spectra.  Such categories were named Poincaré $\infty$-categories in \cite{Ktheory1}. In this section, we review the necessary mathematical background to be used in \S\ref{subsection:surgery}. Our goal is to present a self-contained collection of results most relevant for practical work in  $\rL$-theory, and to get a sense for the main results in the subject. This section may also serve as a point of entry for readers seeking a more detailed exposition, for which we recommend the treatments in \cite{Ktheory1,Ktheory2,Ktheory3}.
 The readers who are familiar with algebraic $\rL$-theory can skip ahead to \S\ref{subsection:surgery} to access the main results.

\subsection{Quadratic Forms and Poincaré Objects}\label{subsection:Ltheory}
One of the central themes of this section is to familiarize the reader with how to categorify the notions of bilinear and quadratic forms by replacing modules with categories and abelian groups with spectra.

\begin{definition}
    Let $\cC$ be a stable $\infty$-category. A \textit{bilinear functor} on $\cC$ is a functor
    \begin{equation}
        B:\cC^{\op} \times \cC^{\op}\to \Sp
    \end{equation}
    such that for every object $C \in \cC$ the functors 
    \begin{equation}
       D\mapsto B(C,D), \quad D \mapsto B(D,C)
    \end{equation}
    are exact functors from $\cC^{\op}$ to $\Sp$.\footnote{By exact functor we mean that it takes zero objects to zero objects and fiber sequences to fiber sequences.}
\end{definition}

\begin{definition}
    For an object $X\in \cC$, and every object $Y\in \cC$, let $\mathrm{Mor}_\cC(Y,X)$ denote the spectrum formed from the sequence of mapping spaces $\{\Hom_{\cC}(Y,\Sigma^n X)\}_{n\geq 0}$. We say that a functor $F: \cC^{\op}\to \Sp$ is \emph{representable} if it maps $Y \mapsto \mathrm{Mor}_\cC(Y,X)$.
\end{definition}

\begin{definition}\label{def:bilinear}
    Consider a functor $\mathbf{D}: \cC^{\op}\to \cC^{}$, which we will refer to as a \textit{duality functor}. We say that a bilinear functor $B$ is \textit{representable} if for all $X\in \cC$ the functor $Y \mapsto B(X,Y)$ is representable. In this case, we will denote $B(X,Y)=\rMor(Y,\bbD(X))$.
\end{definition}
 In the zeroth space of $B(X,\bbD(X))\simeq B(\bbD(X),X)$, the identity map on $\bbD(X)$ determines a point, and therefore gives rise to a morphism $\varepsilon_X: X\to \bbD^2(X)$.

\begin{definition}
    Let $\cC$ be a stable $\infty$-category. A  \textit{nondegenerate bilinear functor} on $\cC$ is a bilinear functor that is representable and the canonical map $\varepsilon_X$ is an equivalence for all $X\in \cC$.
\end{definition}

Consider a functor $Q:\cC^{\op}\to \Sp$ which sends zero $k$-morphisms to zero $k$-morphisms. In this case we call the functor ``reduced''. 
\begin{lemma}
    There is a map from $Q(X)$ to the homotopy $\Z/2$-fixed points of $B(X,X)$.
\end{lemma}
\begin{proof}
    Because $\cC$ is stable, for $X,Y\in \cC$ we have the maps 
\begin{equation}
    Q(X)\oplus Q(Y)\overset{}{\rightarrow} Q(X\oplus Y), \quad Q(X\oplus Y) \overset{}{\rightarrow}Q(X)\oplus Q(Y)
\end{equation}
obtained by applying $Q$ to the inclusion and projection maps to each factor. Furthermore we have a map
\begin{equation}
     Q(X)\oplus Q(Y) \rightarrow Q(X\oplus Y)\rightarrow  Q(X)\oplus Q(Y)\,, 
\end{equation}
where the composition is given by applying $Q$ to the map $\id_X\oplus \id_Y:X\oplus Y{\longrightarrow} X\oplus Y$. If  $Q$ is reduced then the composition is the identity, and hence $Q(X)\oplus Q(Y)$ is a direct summand of $Q(X\oplus Y)$. We consider a decomposition $Q(X\oplus Y)\simeq Q(X)\oplus Q(Y) \oplus B(X,Y)$ for some functor $B:\cC^{\op}\times \cC^{\op}\to \Sp$ referred to as the \textit{polarization} of $Q$. This mimics the standard fact that a symmetric bilinear form can be written using quadratic forms. Using the map $Q(X)\to Q(X\oplus X)$, and projecting onto the component of $B(X,X)$ gives a map $Q(X)\to B(X,X)$. The functor $B$ is symmetric under the $\Z/2$-action on its arguments, and since the sequence of maps we considered are $\Z/2$-invariant, we get a map $Q(X)\to B(X,X)^{h\Z/2}$, where $B(X,X)^{h\Z/2}$ denotes the homotopy fixed point spectrum of $B(X,X)$.
\end{proof}

\begin{definition}
        Let $\cC$ be a stable $\infty$-category. A \textit{quadratic functor} $Q:\cC^{\op}\to \Sp$ is a functor which is reduced, the polarization $B$  is bilinear, and the functor $X\to \mathrm{Fib}(Q(X)\to B(X,X)^{h\Z/2}) $ is exact.
\end{definition}

\begin{example}\label{ex:quadratic}
    Suppose $Q:\cC^{\op}\to \Sp$ is given by $Q(X) = B(X,X)^{h\Z/2}$. Then the polarization of $Q(X\oplus Y)$ is 
    \begin{equation}
       \widetilde{B}(X,Y) = (B(X,Y) \oplus B(Y,X))^{h\Z/2}\simeq B(X,Y).
    \end{equation}
    Furthermore, the canonical map $Q(X)\to \widetilde{B}(X,X)^{h\Z/2}$  is an equivalence and therefore $Q$ is a quadratic functor. In fact, any quadratic functor $Q$ arises as an extension of $B(X,X)^{h\Z/2}$ by an exact functor.
\end{example}

\begin{definition}
    Let $\cC$ be a stable $\infty$-category equipped with a quadratic functor $Q:\cC^{\op}\to \Sp$, and  let $B$ be a symmetric bilinear functor on $\cC$ which is the polarization for $Q$. $Q$ is a \textit{nondegenerate quadratic functor} if $B$ is nondegenerate.
\end{definition}
 In the case where $Q$ is nondegenerate,  $\bbD_Q:\cC^{\op}\to \cC$ implements an equivalence of $\infty$-categories, with $B(X,Y)=\rMor(X,\bbD_Q(Y))$. We now consider the category $(\cC,Q)$ where $Q$ is nondegenerate, and work up to defining $\rL$-groups for this category.

\begin{definition}
   Let $\cC$ be a stable $\infty$-category  with nondegenerate quadratic functor $Q$:
   \begin{itemize}
       \item A \textit{quadratic object}  of $(\cC,Q)$ consists of a pair $(X,q)$ where $X\in \cC$ and $q$ is a point in $\Omega^\infty Q(X)$. \item $(X,q)$ is moreover a \textit{Poincaré object} if the map $X\to \bbD_Q(X)$ determined by $q$ is invertible.
   \end{itemize}
\end{definition}
\noindent One can view $Q$ as a functor that assigns a spectrum of quadratic forms to an object $X\in \cC$, and a quadratic object is just an object in $\cC$ with a quadratic form. 

We now make some important remarks regarding defining the quadratic functor via the formula $Q(X) = B(X,X)_{h \Z/2}$, i.e. by taking homotopy $\Z/2$ coinvariants. One can compute that the polarization $B'$ associated to $Q(X)$ in this case is given by $B'(X,Y)  = B(X,Y)$. The map $Q(X)\to B(X,Y)^{h\Z/2}$ can be identified with the norm map $B(X,X)_{h\Z/2}\to B(X,Y)^{h\Z/2}$. If 2 is invertible in the field we are working with, then the norm map is a homotopy equivalence of spectra. To prime the reader for the remainder of this section, we will work with the coinvariants when defining quadratic $\rL$-theory for rings.

\subsection{$\rL$-groups Classify Poincaré Objects}
We introduce an equivalence condition  on Poincaré objects referred to as ``cobordism''. 
We then define the group $\rL_0(\cC,Q)$
 which classifies Poincaré objects up to cobordism. For a pair  $(\cC,Q)$, we further introduce the 
$\rL$-theory space $\rL(\cC,Q)$, such that
$\pi_0(\rL(\cC,Q)) = \rL_0(\cC,Q)$. We also define the higher $\rL$-groups in this section.

\begin{definition}\label{def:cobordism}
    Let $\cC$ be a stable $\infty$-category with a nondegenerate quadratic functor $Q:\cC^{\op} \to \Sp$. Let $(X,q)$ and $(X',q')$ be two Poincaré objects of $(\cC,Q)$. A cobordism from $(X,q)$ to $(X',q')$ consists of the following data:
    \begin{enumerate}
        \item An object $Y\in \cC$ equipped with maps $\alpha:Y\to X$ and $\alpha':Y\to X'$.
        \item A path $\gamma$ in the space $\Omega^\infty Q(Y)$ between the images of $q$ and $q'$.
        \item A  diagram 
\begin{equation*}
{\begin{tikzcd}[row sep=small]
    &Y \arrow[ld,"\alpha",swap] \arrow[rd,"\alpha'"] \arrow[dd]& \\
    X\arrow[dd]& &X'\arrow[dd]\\
    & \bbD_Q(Y) & \\
   \bbD_Q(X) \arrow[ur,"\bbD_Q(\alpha)"] &  &\bbD_Q(X') \arrow[ul,"\bbD_Q(\alpha')",swap]
\end{tikzcd}}
\end{equation*}
        that commutes up to homotopy, determined by the path $\gamma$.

        The composition of maps
         \begin{equation}\label{eq:null}
            \mathrm{Fib}(\alpha) \to Y \to X' \to \bbD_Q(X')\xrightarrow{\bbD_Q(\alpha')} \bbD_Q(Y)\,,
        \end{equation}
        is nullhomotopic; due to the commutativity of the diagram, the nullhomotopy can be seen by traversing along the left square from $\Fib(\alpha)$ to $\bbD_Q(Y)$.
       \item Equation \eqref{eq:null} gives a map of fibers  $u: \mathrm{Fib}(\alpha)\to \mathrm{Fib}(\bbD_Q(\alpha'))$, and we require that $u$ is  invertible.
    \end{enumerate}
\end{definition}

Adopting the terminology from topology, if two Poincaré objects related by cobordism then they are ``cobordant''.

\begin{proposition}
    Cobordism is an equivalence relation on Poincaré objects of the category $(\cC,Q)$.
\end{proposition}

\begin{proof}
    We first show that cobordism is reflexive, i.e. it is possible to construct a cobordism from $(X,q)$ to itself. This can be done by taking $Y = X$, the maps $\alpha$ and $\alpha'$ to be the the identity maps, and the path $\gamma$ to be the constant path in $\Omega^{\infty}Q$. We now show the symmetric condition. Suppose there is a cobordism between two Poincaré objects, going from $(X,q)$ to $(X',q')$. By Definition \ref{def:cobordism} there is a span 
\begin{equation*}
{\begin{tikzcd}[row sep=small]
    &Y \arrow[ld,"\alpha",swap] \arrow[rd,"\alpha'"] & \\
    X& &X'\,,\\
\end{tikzcd}}
\end{equation*}
    and there is a canonical map $u:\mathrm{Fib}(\alpha)\to \mathrm{Fib}(\bbD_Q(\alpha'))$\,. 
    A cobordism from $(X',q')$ to $(X,q)$ would involve a map $v: \mathrm{Fib}(\alpha') \to \mathrm{Fib}(\bbD_Q(\alpha))$, which we would like to show is an equivalence, given $u$. Let us note the fact that $\mathrm{Fib}(\bbD_Q(\alpha')) = \bbD_Q( \mathrm{coFib}(\alpha'))$, so we can change the target of the maps $u$ and $v$. Since we are working in a stable category, one can equivalently consider the suspension of $v$, which is a map \begin{equation}
        \Sigma(v): \coFib(\alpha')\to \bbD_Q (\Fib(\alpha)).
    \end{equation}
    We observe that $\Sigma(v)$ is isomorphic to $\bbD_Q(u)$, the latter which implements an equivalence. Finally we show that cobordisms are transitive. Consider three Poincaré objects $(X_1,q_1)$, $(X_2,q_2)$, and $(X_3,q_3)$ and the
     the following spans:
    \begin{equation}
         \begin{tikzcd}[row sep=small]
        &Y \arrow[ld,"\alpha_{1}",swap] \arrow[rd,"\alpha_2"]& & Y'\arrow[ld,"\beta_1",swap] \arrow[rd,"\beta_2"] \\
        X_1& &X_2 & &X_3\,,
    \end{tikzcd}
    \end{equation}
   along with a path $\gamma_1$, connecting the image of $q_1$ and $q_2$, and a path $\gamma_2$ connecting the image of $q_2$ and $q_3$ in the spaces $\Omega^\infty Q(Y)$ and $\Omega^\infty Q(Y')$ respectively. We take the pullback of $Y$ and $Y'$ mapping to $X_2$, which yields the diagram:
   \begin{equation}
         \begin{tikzcd}[row sep=small]
         && Y\times_{X_2} Y' \arrow[ld,"\phi_1",swap] \arrow[rd,"\phi_2"] &&\\
        &Y \arrow[ld,"\alpha_{1}",swap] \arrow[rd,"\alpha_2"]& & Y'\arrow[ld,"\beta_1",swap] \arrow[rd,"\beta_2"] \\
        X_1& &X_2 & &X_3\,.
    \end{tikzcd}
   \end{equation}
   We will take $\kappa_1= \alpha_1\circ \phi_1$ and $\kappa_2 = \beta_2\circ \phi_2$ to implement the cobordism between $X_1$ and $X_3$, and concatenate the paths $\gamma_1$ and $\gamma_2$ in the space $\Omega^\infty Q(Y\times_{X_2}Y')$. To show indeed we have a cobordism, we must prove that the map $u:\Fib(\kappa_1) \to \Fib(\bbD_Q(\kappa_2))$ is an isomorphism. Given that there are two cobordisms $(X_1,q_1)$ to $(X_2,q_2)$ and  $(X_2,q_2)$  to  $(X_3,q_3)$, this means we have two isomorphisms $\Fib(\alpha_1) \to \Fib(\bbD_Q(\alpha_2))$ and $\Fib(\beta_1) \to \Fib(\bbD_Q(\beta_2))$. Then we observe that the map $u$ fits into the following diagram of fiber sequences
   \begin{equation}
   \begin{tikzcd}
          \Fib(\beta_1) \arrow[r] \arrow[d]& \Fib(\kappa_1)\arrow[r] \arrow[d,"u"]& \Fib(\alpha_1) \arrow[d] \\
        \Fib(\bbD_Q(\beta_2)) \arrow[r]& \Fib(\bbD_Q(\kappa_2))\arrow[r]& \Fib(\bbD_Q(\alpha_2))\,,
   \end{tikzcd}
   \end{equation}
   hence we see that $u$ is an isomorphism.
\end{proof}

\begin{definition}
    The set of cobordism classes of Poincaré objects in $(\cC,Q)$ is given by $\rL_0(\cC,Q)$. 
\end{definition}

\begin{definition}\label{def:lagrange}
    Let $(X,q)$ be a Poincaré object. A Lagrangian $L$ in $(X,q)$ is a map $r:L\to X$, a nullhomotopy of the image of $q$ in $\Omega^\infty Q(L)$, and an equivalence $L \overset{\sim
    }{\rightarrow} \bbD_Q(X/L)$.
\end{definition}

In the language of spans, one should view a Lagrangian as giving a bordism from a  zero Poincaré object into $(X,q)$, with the map $r:L\to X$ realizing one side of the span. 
 
 Given two Poincaré objects $(X_1,q_1)$ and $(X_2,q_2)$, one can take direct sums $(X_1\oplus X_2,q_1\oplus q_2)$, where the direct sum $q_1\oplus q_2$ is the image of $(q_1,q_2)$ under the map $Q(X_1)\oplus Q(X_2) \to Q(X_1 \oplus X_2)$. The direct sum operation  makes $\rL_0(\cC,Q)$ into a commutative monoid, but it is in fact an abelian group
\begin{lemma}
    $\rL_0(\cC,Q)$ has the structure of an abelian group under direct sum, with the inverse (up to homotopy) of a Poincaré object $(X,q)$ given by $(X,-q)$.
\end{lemma}
\begin{proof}
    We show that the direct sum $(X\oplus X,q\oplus -q)$ is cobordant to the zero Poincaré object. It suffices to show that there is a Lagrangian $r:L \to X\oplus X$. We take $L = X$, $r$ to be the diagonal map, and any path from the sum of $q$ and $-q$ in $\Omega^\infty Q(L)$ to the base point. The equivalence in Definition \ref{def:lagrange} is satisfied because in this case $X \to \bbD_Q(X)$ is already an equivalence since $q$ was assumed to be nondegenerate.  
\end{proof}

\subsubsection{Higher $\rL$-groups}\label{subsub:higherL}
In what follows, we will quote only the necessary facts to be able to study algebraic $\rL$-theory for rings in the next section. We will defer the full details to the notes of Lurie \cite{lurie287x}. 

Drawing on the ideas of $\rK$-theory, one could contemplate developing a $\rL$-theory space $\rL(\cC,Q)$, whose homotopy groups give higher $\rL$-groups. It is particularly sensible to take $\rL(\cC,Q)$ to be the classifying space (or geometric realization) of the simplicial space of Poincaré objects in $(\cC,Q)$, which we denote $\cP(\cC,Q)_*$.

\begin{definition}
    Let $\mathbf{\Delta}$ be the category of combinatorial simplicies with objects $\Delta^0,\Delta^1,\ldots$, given by nonempty finite linearly ordered sets, and morphisms given by order preserving functions $f:\Delta^m\to \Delta^n$ exhibiting how one simplex sits inside another. A \textit{simplicial space} is a functor $X:\mathbf{\Delta}^{\op}\to \mathbf{Spaces}$. Let the $n$th simpicial space of $X$ be denoted by $X(\Delta^n)$.
\end{definition}
\noindent Concretely, a simplicial space is a sequence of spaces equipped with face and degeneracy maps, encoding how simplices of varying dimensions are assembled and related in a coherent, combinatorial manner. If $K$ is a simplicial set, i.e. a functor $K:\mathbf{\Delta}^{\op}\to \mathbf{Sets}$, then $X$ determines a functor such that
\begin{equation}
    K \mapsto \varprojlim _{s:\Delta^n \to K} X(\Delta^n)\,.
\end{equation}
We identify $X(K)$ with the space of maps from $K$ to $X$ in the $\infty$-category of simplicial spaces.

The simplicial space $\cP(\cC,Q)$ is constructed as follows. Given an $\infty$-category, one can look at its simplicial structure in each dimension by using a functor from ordered sets. Let $\cF_n$ denote the collection of nonempty subsets of the set $\{0,1,\ldots,n\}$, and  regard $\cF_n$ as a partially ordered set under inclusion.
Let $\cC_{[n]} = \Fun(\cF^{\op}_n,\cC)$, and $Q_{[n]}: \cC^{\op}_{[n]}\to \Sp$ be a quadratic functor given by the  formula
\begin{equation}
    Q_{[n]}(X) = \varprojlim _{S \in \cF_n} Q(X(S)), \quad X\in \cC^{\op}_{[n]}\,.
\end{equation}
 For a finite ordered set $[n] = \{0,1,\ldots,n\}$,  a map $f:[m]\to[n]$ induces a functor $f^*:\cC_{[n]} \to \cC_{[m]}$ due to the  construction $[n]\to \cC_{[n]}$ being contravariantly functorial. From $f^*$ one can establish the following.
\begin{proposition}
    Every Poincaré object of $(\cC_{[n]},Q_{[n]})$ is mapped to a Poincaré object of $(\cC_{[m]},Q_{[m]})$ under the functor $f^*:\cC_{[n]} \to \cC_{[m]}$.
\end{proposition}
For each $n\geq 0$ there is a classifying space  for Poincaré objects of $(\cC_{[n]},Q_{[n]})$ given by $\cP(\cC,Q)_{n}$. The map of finite sets $f:[n]\to [m]$ induces a map of spaces $\cP(\cC,Q)_{n}\to \cP(\cC,Q)_{m}$. This gives the space $\cP(\cC,Q)_{*}$ the structure of a simplicial space.

In general it is not easy to compute the homotopy groups of a geometric realization i.e. $\rL(\cC,Q)$ solely from knowledge of the homotopy groups of $\cP(\cC,Q)_*$. However, it is a nontrivial fact that:
\begin{theorem}
    The simplicial space $\cP(\cC,Q)_*$ satisfies the Kan condition.
\end{theorem}
\noindent If a simplicial space $X$ satisfies the Kan condition then it is a Kan complex, which mean that every horn in $X$ can be filled to a full simplex. This ensures that the classifying space of $X$ does not lose information about the simplicial space. 

If $Y$ is a simplicial space that satisfies the Kan condition, then the homotopy groups of its classifying space $|Y|$ is computed in an algorithmic way using simplicial methods.  Suppose $K$ is a simplicial set, and  $K_0$ is a simplicial subset of $K$. Let $Y(K,K_0)$ denote the homotopy fiber of the map $Y(K)\to Y(K_0)$. Consider the space $Y(\Delta^n,\partial \Delta^n)$ where $\partial \Delta^n$ denotes the boundary simplices of $\Delta^n$. Consider a subset $K \subseteq \partial \Delta^{n+1}$ obtained by removing the interiors of two faces in $\partial \Delta^{n+1}$. 
\begin{proposition}
    There is a canonical bijection 
\begin{equation}
    Y(\partial \Delta^{n+1},K) \to Y(\Delta^n, \partial \Delta^n)\times Y(\Delta^n, \partial \Delta^n)\,,
\end{equation}
and a pair of objects $\lambda, \lambda' \in Y(\Delta^n, \partial \Delta^n)$, determine the same object in $\pi_n |Y|$ if and only if the corresponding element in $Y(\partial \Delta^{n+1},K)$ lifts to a point in $Y(\Delta^{n+1})$.
\end{proposition}

 This is equivalent to the existence of a cobordism between $\lambda$ and $\lambda'$. Hence, $Y(\Delta^n,\partial \Delta^n)$ represents the classifying space for Poincaré objects $(X,q)$ of  $\cC_{[n]}$ such that $X(S)\simeq 0$ for all $S \subseteq [n]$. One can see this by the fact that a Poincaré object of $\cC_{[n]}$ is a pair of Poincaré objects in $\cC_{[n-1]}$ with a cobordism between them. Since $X(S)\simeq 0$ for all $S \subseteq [n]$, $X$ is determined by a single object $X([n])\in \cC$. Furthermore we have 
\begin{equation}
    Q_{[n]}(X) =   \varprojlim _{S \in \cF_n} Q(X(S)) =\begin{cases}
        Q(X([n])), \quad &S=[n]  \\
        0,  &\text{else}\,.
    \end{cases}
\end{equation}
The diagram that describes this limit is parametrized by a partially ordered set of faces of an $n$-simplex, and is trivial on every proper face. Thus, the limit is given by $\Omega^n Q(X([n]))$. 
If $Y = \cP(\cC,Q)_*$ with $|Y| = \rL(\cC,Q)$, then the space $Y(\Delta^n,\partial \Delta^n)$ is the classifying space for Poincaré objects of $(\cC_{[n]},Q_{[n]})=(\cC,\Omega^n Q)$, i.e. $Y(\Delta^n,\partial \Delta^n )  = \rL(\cC,\Omega^n Q)$.

The discussion thus far is summarized by the following theorem:


\begin{theorem}
    The abelian group $\pi_n(\rL(\cC,Q))$ is canonically isomorphic to $\pi_0(\rL(\cC,\Omega^n Q))$.
\end{theorem}

\begin{example}\label{example:shiftq}
    As an illustrative way to see what the shifted quadratic functor $\Omega^n
    Q$ means, consider a cobordism between two zero Poincaré objects $(X,q)$ and $(X',q')$. The cobordism between these two objects is determined by an object $L\in \cC$ with a point $p\in \Omega (\Omega^{\infty} Q(L))$, i.e. a point in the based loop space of $\Omega^{\infty} Q(L)$. But this is just a path from the trivial $q$ to the trivial $q'$.  $p$ is furthermore required to induce an equivalence $L \to \Omega \bbD_Q(L)$, which means an equivalence only after shifting down by a degree after applying the duality functor. We therefore say that $(L,p)$ is a Poincaré object in $\cC$ with respect to a shifted quadratic functor $\Omega Q$.
\end{example}

\subsection{L-theory of Rings and Surgery}
We now specify to a particular choice of $\cC$ that is useful for classifying Pauli stabilizer codes: this is the category $(\cD^\perf(R),Q)$ of perfect complexes over $R$ equipped with a quadratic functor $Q$.

\begin{definition}
    Let $R$ be an associative ring. The stable $\infty$-category of perfect chain complexes $D^{\perf}(R)$ consists of:
    \begin{itemize}
        \item Objects given by bounded chain complexes of finitely generated projective $R$-modules,
        \item 1-morphisms given by maps $f:P_*\to Q_*$ between chain complexes,
        \item Higher morphisms given by chain homotopies.
    \end{itemize}
\end{definition}

\begin{example}
    The idea of Poincaré objects in a stable $\infty$-category defined over a ring can be used to recover Poincaré duality which is enjoyed by manifolds. Let $M$ be a manifold of dimension $n$. The singular cochain complex $C^*(M,\Z)$ is an object of $D^{\perf}(\Z)$.\footnote{Technically $D^{\perf}(\Z)$ has objects which are perfect chain complexes and this example uses a cochain complex. But since $M$ is finite dimensional, perfect chain complexes dualize to perfect cochain complexes.} We define a bilinear functor $B:D^{\perf}(\Z) \times D^{\perf}(\Z)\to \Sp$ by the formula
    \begin{equation}
        B(P_*,Q_*) = \rMor(P_* \otimes Q_*,\Z[-n])\,,
    \end{equation}
     where $\Z[-n]$ denotes a chain complex consisting of a single group $\Z$ in homological degree $-n$.
    We then take the quadratic functor $Q:D^{\perf}(\Z)\to \Sp$ to be given by $Q(P_*)=B(P_*,P_*)^{h\Z/2}$.
    The intersection pairing:
    \begin{equation}
        C^*(M;\Z) \otimes C^*(M;\Z) \overset{\cup}{\rightarrow}C^*(M;\Z) \overset{[M]}{\rightarrow} \Z[-n]\,,
    \end{equation}
    determines a point $q_M\in \Omega^\infty Q(C^*(M,\Z))$. Thus Poincaré duality is equivalent to the statement that $(C^*(M,\Z),q_M)$ is a Poincaré object of $(D^\perf(\Z),Q)$.
\end{example}

The ring $R$ which is relevant for translation invariant Pauli stabilizer codes comes with an involution, as explained in \S\ref{subsection:introtopauli}. Therefore, any left $R$-module $M$ can be regarded as a right $R$-module by using the formula
\begin{equation}
    m r = \overline{r} m, \quad r\in R,\,m\in M.
\end{equation}

The duality map $\bbD$ used in \S\ref{subsection:Ltheory} establishes a contravariant equivalence of $D^{\perf}(R)$ with itself. The action of $\bbD$ on a complex $P_*$ applies the $R$-linear dual $\hom_R(-,R)$ termwise. In parallel with Definition \ref{def:bilinear}: a bilinear functor defined on $D^{\perf}(R)$ is given by the formula $B(P_*,Q_*) = \rMor(P_*,\bbD(Q_*))$, and is symmetric in its entries due to the fact that the involution on $R$ squares to the identity.


\begin{definition}
    Let $R$ be a ring with involution and let $Q: D^{\perf}(R)\to \Sp$  be a quadratic functor given by $Q(X)= B(X,X)_{h\Z/2}$. Then for every integer $n$, define the group 
   \begin{equation}
      \rL_n(R):=\rL_n(D^{\perf}(R),Q) = \rL_0(D^{\perf}(R),\Omega^n Q)\,.
   \end{equation}
\end{definition}

In the same that was explained in Example \ref{example:shiftq}, one can view Poincaré objects of $(D^{\perf}(R),\Omega^n Q)$ to comprise of an object $P_*\in D^{\perf}(R)$ that is equipped with a nondegenerate quadratic form such that there is an equivalence  $P_* \simeq \Omega^n \bbD(P_*)$ under duality.

In analogy with  $\rL$-theory for manifolds, one gains substantial leverage in the study of higher $\rL$-groups for rings by exhibiting cobordisms from general Poincaré objects to ones concentrated in the middle degree. We will refer to such cobordisms as ``surgeries''.
Let us explain surgery in a general setting; as usual, take $\cC$ to be a stable $\infty$-category and $Q$ to be a nondegenerate quadratic functor with polarization $B$, and associated duality functor $\bbD_Q$. We will explain what it means to perform surgery on a Poincaré object $(X,q)$ along a map $\sigma:W\to X$.
We have observed before in Equation \eqref{eq:null}, that the composition of maps 
\begin{equation}
    W \overset{\sigma}{\rightarrow} X \rightarrow \bbD(X)\rightarrow \bbD(W)
\end{equation}
is canonically nullhomotopic. However we note that the maps 
\begin{equation}\label{eq:triangle}
    W \overset{\sigma}{\rightarrow} X \overset{\beta}{\rightarrow} \bbD(X)
\end{equation}
do not in general form a fiber sequence, and the failure to be a fiber sequence is  captured by an object $X_\sigma := \Fib(\beta)/W$, along with a choice of nullhomotopy for $q|_{W}$.

\begin{lemma}\label{lem:fibersequence}
    Given a fiber sequence $W\to X\to Y$ in a stable $\infty$-category $\cC$, for every quadratic functor $Q$ on $\cC$ there is a fiber sequence
    \begin{equation}\label{eq:fibforQ}
        Q(Y)\rightarrow Q(X) \rightarrow Q(W) \times_{B(W,W)} B(X,W)\,.
    \end{equation}
\end{lemma}

\begin{proof}
    We observe that since the bilinear form $B$ is exact in each variable, for a fiber sequence $W\to X\to Y$ we obtain a diagram:
    \begin{equation}
         \begin{tikzcd} 
        B(Y,Y) \arrow[r] \arrow[d]& B(Y,X) \arrow[r] \arrow[d]& B(Y,W) \arrow[d]\\ 
         B(X,Y) \arrow[r] \arrow[d]& B(X,X) \arrow[r] \arrow[d]& B(X,W)\arrow[d]\\
          B(W,Y) \arrow[r] & B(W,X) \arrow[r] & B(W,W)\,,
    \end{tikzcd}
    \end{equation}
   where each row and column are fiber sequences. Thus, we have a fiber sequence 
    \begin{equation}\label{eq:fiber}
       B(Y,Y) \to B(X,X) \to B(X,W) \times_{B(W,W)} B(W,X)\,.
    \end{equation}
    Using the fact that 
    \begin{equation}
        B(X,W) \times_{B(W,W)} B(W,X) \simeq (B(X,W)\times B(W,X)) \times_{B(W,W)\times B(W,W)} B(W,W)
    \end{equation}
    we can substitute this into Equation \eqref{eq:fiber} and take homotopy fixed points with respect to $\Z/2$. This gives the fiber sequence
    \begin{equation}
        Q(Y)\rightarrow Q(X) \rightarrow Q(W) \times_{B(W,W)} B(X,W)\,
    \end{equation}
    as desired.
\end{proof}

\begin{lemma}\label{lem:sigmaquadratic}
There is a lift from $q\in \Omega^\infty Q(X)$ to $q_\sigma \in \Omega^\infty Q(X_\sigma)$.  The object $(X_\sigma,q_\sigma)$ is thus a quadratic object in $(\cC,Q)$.
\end{lemma}
\begin{proof}
Using Lemma \ref{lem:fibersequence}, we can construct a fiber sequence from Equation \eqref{eq:triangle} as follows: 
\begin{equation}
    Q(X_\sigma) \rightarrow  Q(\Fib(\beta))\rightarrow Q(W) \times_{B(W,W)} B(\Fib(\beta),W)\,.
\end{equation}
    Since $q\in \Omega^\infty Q(X)$, it determines a point in $\Omega^\infty B(X,W)$, where this space is viewed as the cofiber of $Q(W)\to Q(X)$. Such a point classifies the maps $\beta:X \to \bbD W$. By composing with $\Fib(\beta)\to X$, such a 
map is canonically nullhomotopic. Therefore the restriction $q|_{\Fib(\beta)}$ has a trivial image in $\Omega^\infty(Q(W) \times_{B(W,W)} B(\Fib(\beta),W))$, and so lifts to a point $q_\sigma \in \Omega^\infty Q(X_\sigma)$.
\end{proof}

\begin{proposition}\label{prop:surgery}
    Let $(X,q)$ be a Poincaré object of $(\cC,Q)$. Given a map $\sigma: W \to X$ and a nullhomotopy of $q|_W$, then $(X_\sigma,q_\sigma)$ constructed by surgery along $\sigma$ is a Poincaré object of $(\cC,Q)$.
\end{proposition}

\begin{proof}
    Since Lemma \ref{lem:sigmaquadratic} proved that $(X_\sigma,q_\sigma)$ is quadratic, it suffices to show that $q_\sigma$ provides an isomorphism from $X_\sigma$  to $\bbD(X_\sigma)$. Given a triangle in a stable $\infty$-category,
    applying the duality functor gives another triangle. Therefore we get a map between triangles
    \begin{equation}\label{eq:triangles}
    \begin{tikzcd}
         W \arrow[r,"\sigma"] \arrow[d,"\cong"]& X \arrow[r,"\beta"] \arrow[d]& \bbD (W) \arrow[d,"\cong"] \\
         \bbD^2(W) \arrow[r]& \bbD(X) \arrow[r]& \bbD(W)
    \end{tikzcd}
    \end{equation}
    where the bottom line is obtained by applying $\bbD$ to the top line. The two vertical maps on the left and right are isomorphism, and the middle map is induced by $q$.
    From the top line, we can extract the object $X_\sigma = \coFib(W\to \Fib(\beta))$. Under duality we see that 
    \begin{align}
        \bbD(X_\sigma) &= \bbD(\coFib(W\to \Fib(\beta))) \\\notag 
           &=\Fib(\bbD(W \to \Fib(\beta)))\\\notag 
           &= \Fib(\bbD(\Fib(\beta))\to \bbD(W))\\\notag 
           & = \Fib(\coFib(\bbD(\beta))\to \bbD(W))\,,
    \end{align}
    and so indeed we get a map  $q_\sigma:X_\sigma\to \bbD(X_\sigma)$ induced by the map on triangles. The cofiber of $q_\sigma$ is obtained by looking at the cofiber of the map on triangles in Equation \eqref{eq:triangles}. But since the outter two maps are equivalences, then the cofiber of $q_\sigma$ is equivalent to the cofiber of $X\to \bbD(X)$. This vanishes because $(X,q)$ is a Poincaré object and thus $q_\sigma$ is an isomorphism.
\end{proof}

In summary, there is a cobordism of $(X,q)$ along $\sigma$ that produces another Poincaré object $(X_\sigma,q_\sigma)$. Since $q_\sigma$ and $q$ have the same image when restricted to $\Fib(X\to \bbD(W))$, this can be used to show that they determine the same element of the abelian group $\rL_0(\cC,Q)$.  

We now explain how to apply surgery along $\sigma$ to create Poincaré objects concentrated near the middle dimensions. 
\begin{definition}
    A spectrum $X$ is \emph{$n$-connective} if $\pi_i(X)$ vanish for $i<n$.  $X$ is called \emph{connective} if its negative homotopy groups vanish. 
\end{definition}

Let $M$ be a module spectrum over $R$, with dual $\bbD(M)$ given by the mapping spectrum $\rMor_R(M,R)$. To be concrete, one can take $M$ to be an element of $D^\perf(R)$.

\begin{definition}
    A module spectrum $M$  over $R$ has \textit{projective height} $\leq n$ if $\bbD(M)$ is $(-n)$-connective. 
\end{definition}

\begin{proposition}\label{prop:height}
    Let $R$ be a ring and $M\in D^\perf(R)$ which is connective and of projective height $\leq 0$. Then $M$ is a direct summand of $R^n$, for some  $n$. \footnote{If $M$ was instead $k$-connective and of projective height $\leq k$ then the analogue of the statement in the proposition is that $M$ is a summand of $\Sigma^k R^n$.}
\end{proposition}

\begin{proof}
    The fact that the negative homotopy groups of $M$ vanish, along with the fact that $M$ is in $D^{\perf}(R)$ can be used to establish that $\pi_0(M)$ is a finitely generated $R$-module. We can choose a surjective on $\pi_0$ map $R^n\to M$, which gives a fiber sequence 
    \begin{equation}
        N \rightarrow R^n \rightarrow M
    \end{equation}
    where $N$ is also connective. This  gives a fiber sequence of mapping spectra
    \begin{equation}
        \begin{tikzcd}
            \rMor_R(M,N) \arrow[r]& \rMor_R(M,R^n) \arrow[r]& \rMor_R(M,M)\,,
        \end{tikzcd}
    \end{equation}
    for which we notice that  $\rMor_R(M,N)$ is also connective as $N$ is connective and $M$ has projective height $\leq 0$. In the long exact sequence of homotopy groups, this fact implies the map $\pi_0(\rMor_R(M,R^n))\to \pi_0(\rMor_R(M,M))$ is surjective. Therefore the identity map $\id_M: M\to M$ lifts to a map $M \to R^n$, and we conclude.
\end{proof}

Let $Q(M) = \rMor_{R\text{-}R}(M\otimes M,R)_{h\Z/2} =B(M,M)_{h\Z/2}$ for the remainder of this section. We will prove the following proposition, which describes the two cases of how  Poincaré objects can be simplified via surgery, in even and odd degree:
\begin{proposition}\label{prop:evenodd}
    Let $(M,q)$ be a Poincaré object of $(D^{\perf}(R),\Sigma^{2k}Q)$. Then $(M,q)$ is cobordant to another Poincaré object $(N,q')$ with $\pi_i(N) = 0$ if $i<k$. Similarly, if  $(M,q)$ be a Poincaré object of $(D^{\perf}(R),\Sigma^{2k+1}Q)$. Then $(M,q)$ is cobordant to another Poincaré object $(N,q')$ with $\pi_i(N) = 0$ if $i<k$.
\end{proposition}

We break up the proof into the following Lemmas, first starting with the case when the number of suspensions is even.
\begin{lemma}\label{lem:pullbackq}
    If $(M,q)$ is a Poincaré object of $(D^\perf(R),\Sigma^{2k}Q)$ and $M = \Sigma^m R$ for $m <k$, then there exists a nullhomotopy for $q$.
\end{lemma}

\begin{proof}
    The quadratic form $q$ determines a point in the zeroth space $\Sigma^{2k}(\bbD(M) \otimes \bbD(M))_{h\Z/2}$, which can be identified with $\Sigma^{2k-2m} R$. But if $m <k$ then the homotopy groups of $\Sigma^{2k-2m} R$ vanishes in non-positive degrees. Therefore there exists a nullhomotopy for $q$. 
\end{proof}

\begin{lemma}\label{lem:surgeryodd}
    Let $(M,q)$ be a Poincaré object of $(D^\perf(R),\Sigma^{2k}Q)$ and let $\nu \in \pi_m(M)$, where $m <k$. There exists another Poincaré object $(M_\sigma,q_\sigma)$ that is cobordant to $(M,q)$, where $\pi_i(M_\sigma)= \pi_i(M)$ for $i<m$, and $\pi_m(M_\sigma)= \pi_m(M)/\nu $.
\end{lemma}
\begin{proof}
    The class $\nu$ determines a map $\sigma: \Sigma^m R \to M$. If $m <k$, by Lemma \ref{lem:pullbackq}, $q|_{\Sigma^{m}R}$ is nullhomotopic. By choosing a nullhomotopy, Proposition \ref{prop:surgery} gives a surgery along $\sigma$ to another Poincaré object $(M_\sigma,q_\sigma)$.  Consider a map $\beta: M \to \Sigma^{2k} \bbD(\Sigma^m R)$, determined by the choice of nullhomotopy for $q|_{\Sigma^m R}$. From the fiber sequence  \begin{equation}
    \begin{tikzcd}
        \Fib(\beta) \arrow[r]& M \arrow[r]& \Sigma^{2k-m}R
    \end{tikzcd}
    \end{equation}  
    we obtain a long exact sequence in homotopy:
    \begin{equation}
    \begin{tikzcd}
         \ldots \arrow[r]&\pi_{i+1}(\Sigma^{2k-m} R) \arrow[r]& \pi_i(\Fib(\beta)) \arrow[r]& \pi_i(M) \arrow[r]& \pi_i(\Sigma^{2k-m}R) \arrow[r]&\ldots \,.
    \end{tikzcd}
    \end{equation}
    Since the ring $R$ is connective, $\pi_i(\Sigma^{2k-m}R)$ is trivial if $i+m < 2k$, and  $\pi_{i+1}(\Sigma^{2k-m}R)$ is trivial if $i+m +1< 2k$. Thus $\pi_i(\Fib(\beta))  = \pi_i(M)$ for $i\leq k$. To compare $\pi_i(M)$ with $\pi_i(M_\sigma)$ we consider the long exact sequence:
    \begin{equation}
         \begin{tikzcd}
         \ldots \arrow[r]&\pi_{i}(\Sigma^{m} R) \arrow[r]& \pi_i(\Fib(\beta)) \arrow[r]& \pi_i(M_\sigma) \arrow[r]& \pi_{i-1}(\Sigma^{m}R) \arrow[r]&\ldots \,,
    \end{tikzcd}
    \end{equation}
    which we have because $(M,q)$ and $(M_\sigma,q_\sigma)$ are cobordant.
    In the same way as earlier, we conclude that $\pi_i(\Fib(\beta)) = \pi_i(M_\sigma)$ for $i<m$, and thus $\pi_i(M_\sigma) = \pi_i(M)$ for $i<m$. In the case when $i=m$, we see that $\pi_m(M_\sigma)$ is the quotient of $\pi_m(\Fib(\beta)) = \pi_m(M)$ by the submodule generated by $\nu$.
\end{proof}

The case when the number of suspensions is odd follows analogously.

\begin{lemma}\label{lem:surgeryother}
      Let $(M,q)$ be a Poincaré object of $(D^\perf(R),\Sigma^{2k+1}Q)$ and let $\nu \in \pi_m(M)$, where $m \leq k$. There exists another Poincaré object $(M_\sigma,q_\sigma)$ that is cobordant to $(M,q)$, where $\pi_i(M_\sigma)= \pi_i(M)$ for $i<m$, and $\pi_m(M_\sigma)= \pi_m(M)/\nu $.
\end{lemma}
\begin{proof}
    We note that if $M= \Sigma^{m} R$ then the space $\Sigma^{2k+1}Q(M) = (\Sigma^{2k+1-2m} R)_{h\Z/2}$ is connected for $k\leq m$. By the same logic as the proof of Lemma \ref{lem:surgeryodd}, for a class $\nu \in \pi_m(M)$ and $m\leq k$ we can perform a surgery to another Poincaré object $(M_\sigma,q_\sigma)$.  We consider a map $\beta: M \to \Sigma^{2k+1}\bbD(\Sigma^m R) = \Sigma^{2k+1-m} R$, for which we have the long exact sequence on homotopy:
    \begin{equation}
          \begin{tikzcd}
         \ldots \arrow[r]&\pi_{i+1}(\Sigma^{2k+1-m} R) \arrow[r]& \pi_i(\Fib(\beta)) \arrow[r]& \pi_i(M) \arrow[r]& \pi_i(\Sigma^{2k+1-m}R) \arrow[r]&\ldots \,.
    \end{tikzcd}
    \end{equation}
    Thus, $\pi_i(\Fib(\beta)) =\pi_i(M)$ for $i < 2k-m$. In particular, if $m<k$ then the two groups agree when $i <m$. By looking at the long exact sequence 
     \begin{equation}
         \begin{tikzcd}
         \ldots \arrow[r]&\pi_{i}(\Sigma^{m} R) \arrow[r]& \pi_i(\Fib(\beta)) \arrow[r]& \pi_i(M_\sigma) \arrow[r]& \pi_{i-1}(\Sigma^{m}R) \arrow[r]&\ldots \,,
    \end{tikzcd}
    \end{equation}
    we see that $\pi_i(M_\sigma) = \pi_i(M)$ for $i <m$, and $\pi_m(M_\sigma) = \pi_m(M)/\nu$.
\end{proof}

\begin{proof}[Proof of Proposition \ref{prop:evenodd}]
    We first consider the case involving $\Sigma^{2k}Q$. Since $R$ is connective, we consider the smallest positive $m$ for which $\pi_m(M)$ is nontrivial. As long as $m<k$ we apply Lemma \ref{lem:surgeryodd} repeatedly: for each $i\leq m$ we construct a cobordism to $(M_\sigma,q_\sigma)$ for which $\pi_i(M_\sigma)=0$, tivializing the high values of $i$ first and moving to the low values. The same algorithm works for the case $\Sigma^{2k+1}Q$ by invoking Lemma \ref{lem:surgeryother}.
\end{proof}

\begin{corollary}\label{cor:singledegree}
    Every Poincaré object $(M,q)$ of $(D^{\perf}(R),\Sigma^{2k}Q)$ is cobordant to a Poincaré object $(N,q')$ where $N$ is a direct summand of $\Sigma^k R^n$ for some $n$.
\end{corollary}
\begin{proof}
    Using Proposition \ref{prop:evenodd} we take $N$ to be $k$-connective. Under the equivalence $N \simeq \Sigma^{2k}\bbD(N)$, we see that $N$ has projective height $k$. Proposition \ref{prop:height} can then be used to show that $N$ is a direct summand of $\Sigma^k R^n$ for some $n$.
\end{proof}

\begin{rem}\label{rem:twolayers}
    In the case when $(M,q)$ is a Poincaré object of $(D^{\perf}(R),\Sigma^{2k+1}Q)$, Proposition \ref{prop:evenodd} shows that $(M,q)$ is cobordant to a Poincaré object $(N,q')$ where $N$ is $k$-connective. However, the equivalence $N \simeq \Sigma^{2k+1}\bbD(N)$ only guarantees that $N$ has projective height $(k+1)$. Therefore,  $(M,q)$ cannot be cobordant to a Poincaré object that is only concentrated in a single degree; rather, we can ensure that there are contributions to two degrees after surgery.
\end{rem}


To remain aligned with the physical motivation underlying the formalism introduced in this section, we emphasize that the propositions developed here substantially streamline the manipulation and classification of Poincaré objects. In particular, once a physical system is realized as a Poincaré object within this framework, its classification up to gapped interfaces reduces to the classification of middle-dimensional operators. Proposition \ref{prop:evenodd} has an analogy that we will explore in Construction \ref{construction}, with an eye towards TQFTs.

\subsection{Computing $\rL$-groups}\label{subsection:Lgroupfork}
For our classification purposes, we will need the descriptions for the $\rL$-groups of a field $\mathds{k}$.  In what follows, we describe these groups using the formalism developed in the previous sections.

\begin{proposition}\label{prop:4periodic}
   Let $R$ be an associative ring with involution. The groups $\rL_n(R)$ are four-periodic if 2 is invertible in $R$, and are two-periodic if 2 is not invertible.
\end{proposition}
\begin{proof}
    Since the category $D^\perf(R)$ used to define the above $\rL$-groups is a stable $\infty$-category, the functor $\Sigma$  is an equivalence on  $D^\perf(R)$. The bilinear form $B(M,N) = \rMor_{R\text{-}R}(M\otimes  N,R)$ and satisfies $B(\Sigma M,\Sigma N)\simeq \Sigma^{-2} B(M,N)$. Now we consider homotopy $\Z/2$ coinvariants of $\Sigma^{-2}B(\Sigma M,\Sigma N)$, which is needed to establish the quadratic functor on $D^\perf(R)$. The $\Z/2$-action  on $\Sigma^{-2}$  is given by permuting suspension coordinates, and acts by a sign. However, the square of the action is trivial. Therefore, we have an equivalence $B(\Sigma^2 M, \Sigma^2 N)\simeq \Sigma^{-4} B(M,N)$, and consequently $(D^\perf(R),Q)\simeq (D^\perf(R),\Sigma^{-4}Q)$. In the case when 2 is non-invertible, the action of $\Z/2$ on $\Sigma^{-2}$ already was trivial. By the same line of reasoning, we find that $(D^\perf(R),Q)\simeq (D^\perf(R),\Sigma^{-2}Q)$.
\end{proof}


\begin{proposition}\label{prop:oddLgroups}
    The groups $\rL_{-2k-1}(\mathds{k})$ are trivial.
\end{proposition}
\begin{proof}
By Lemma \ref{lem:surgeryother}  we can reduce to the case where $(M,q)$ is a Poincaré object of $(D^{\perf}(R),\Sigma^{2k+1}Q)$, and $M$ is $k$-connective. We would like to exhibit a Lagrangian for $(M,q)$. Notice that nondegeneracy of $q$ exhibits $\pi_i(M)$ as a $\mathds{k}$-linear dual of $\pi_{2k+1-i}(M)$ for each $i$. But the homotopy groups $\pi_i(M)=0$ unless $i=k$ or $i=k+1$, by Remark \ref{rem:twolayers}. Hence, we have a duality $\pi_{k+1}(M) = \pi_k(M)^\vee$. Let $L=\pi_k(M)$ and $\Sigma^k  L$ denote the module given by a single copy of $L$ placed in degree $k$. Since $\mathds{k}$ is a field, $L$ is a free $\mathds{k}$-module and therefore we may construct a map 
\begin{equation}
    \sigma : \Sigma^k L \rightarrow M
\end{equation}
and we can check there is an equivalence $\Sigma^k L \simeq \Sigma^{2k+1}\bbD(M/\Sigma^k L)$. Furthermore, $\Sigma^{2k+1}Q(\Sigma^k L) \simeq \Sigma^1 ( L \otimes_{\mathds{k}} L)$ is a connected space. Therefore $q|_{\Sigma^k}$ is nullhomotopic and any choice of nullhomotopy exhibits $\Sigma^k L$ as a Lagrangian for $M$.
\end{proof}

\begin{proposition}\label{prop:2mod4}
    The groups $\rL_{-4k-2}(\mathds k)$ are trivial.
\end{proposition}
\begin{proof}
    Let $(M,q)$ be a Poincaré object of $(D^{\perf}(R),\Sigma^{4k+2}Q)$. By Corollary \ref{cor:singledegree} we will take $M = \Sigma^{2k+1} V$ where $V$ is a vector space over $\mathds{k}$. The space $\Sigma^{4k+2}Q(\Sigma^{2k+1} V) = \Sigma^{4k+2}(\Sigma^{-4k-2} \rMor(V\otimes V,k))_{h\Z/2}$, where $\rMor_{\mathds k}(V\otimes V,\mathds k)$ is the $\mathds k$-vector space of symmetric bilinear forms from  $V \otimes V \to \mathds k$. Since the $\Z/2$-action is nontrivial on the suspensions, $\Sigma^{4k+2}Q(\Sigma^{2k+1} V)$ has the interpretation as the vector space of skew-symmetric bilinear forms from  $V \otimes V \to \mathds k$. Moreover, since $(M,q)$ was a Poincaré object, the corresponding skew symmetric bilinear forms should be nondegenerate. We can always find a Lagrangian reduction to a subspace $L \subseteq V$ of half dimensions so that a nondegenerate skew-symmetric bilinear form $b|_{L\times L}$ is trivial. Therefore, $L$ is a Lagrangian of $V$.
\end{proof}

The last $\rL$-group to contend with are the ones in degree $4k$. These are identified with the Witt group 
$\Witt(\Z/p)$. By definition, elements of this group are nondegenerate quadratic spaces over 
$\Z/p$, considered up to the equivalence relation generated by the addition of hyperbolic planes. In other words, two quadratic spaces represent the same class if their direct sum becomes isometric after adjoining a suitable number of hyperbolic summands. We now describe this identification in more detail, and we give a more specific account for when the field is $\Z/2$ at the end.

\begin{definition}
    A quadratic space over a field $\mathds k$ is a pair $(V,q)$ of a finite dimensional vector space $V$ over $\mathds k$ and $q:V\to \mathds k$ is a quadratic form. We say that $q$ is \textit{nondegenerate} if the associated bilinear form is nondegenerate.  
\end{definition}

\begin{definition}
    A \textit{hyperbolic plane} is a quadratic space $H = (\mathds{k}^2,q)$ over $\mathds{k}$ with quadratic form $q(a,b) = ab$.
\end{definition}

\begin{rem}
    Let $(V,q)$ be a nondegenerate quadratic space over $\mathds k$, and let $V_0\subseteq V$ be a subspace such that $(V_0,q|_{V_0})=H$. Then there is a decomposition $(V,q)\simeq H \oplus (V_1,q|_{V_1})$ where $V_1$ is the orthogonal complement of $V_0$.
\end{rem}

\begin{definition}
    Two nondegenerate quadratic spaces $(V,q)$ and $(W,q')$ are stably equivalent if $(V,q)\oplus H^a \cong (W,q')\oplus H^b$ for some  $a,b\in \Z$. The \textit{Witt space} $\Witt(\mathds k)$ is the collection of equivalence classes of stably equivalent  nondegenerate quadratic spaces  over $\mathds k$.
\end{definition}

\begin{proposition}\label{prop:0mod4}
    The map $\varphi:\Witt(\mathds k) \to \rL_{0}(\mathds k)$ is an isomorphism of abelian groups.
\end{proposition}
\begin{proof}
Consider a Poincaré object $(V,q)$ of $(D^\perf(R),Q)$, where we  view $V$ as a 
 chain complex over $\mathds k$ that is concentrated in degree 0. This allows us to view a nondegenerate quadratic space $(V,q)$ as a Poincaré object and  gives a map $\varpi: \Witt(\mathds k)\to \rL_{0}(\mathds k)$, that is surjective due to surgery. We now show that $\varphi$ is injective. Suppose that there exists a Lagrangian in $V$ given by $L$. This gives a self-dual fiber sequence of spectra 
 \begin{equation}
    \begin{tikzcd}
        L \arrow[r,"\alpha"]& V \arrow[r]& \coFib(\alpha)\,,
    \end{tikzcd}
 \end{equation}
 with the self duality of $V$ implemented by $q$. This leads to the following self-dual short exact sequence of vector spaces
\begin{equation}
    0 \to \hom(\mathrm{Im} (\pi_0L), V) \to V \to \hom(\mathrm{Im} V , \pi_0 (\coFib(\alpha)))\,.
\end{equation}
 The self-duality implies that the dimension of $V$ is twice as large as the dimensions of $\hom(\mathrm{Im} (\pi_0L), V)$ and $\hom(\mathrm{Im} V , \pi_0 (\coFib(\alpha)))$. As the map  $\hom(\mathrm{Im} (\pi_0L), V)\to V$ factors through $L$, we see that $q|_{\hom(\mathrm{Im} (\pi_0L), V)}=0$. The self-duality then implies that $V$ therefore decomposes into direct sum of hyperplanes, and is trivial in $\Witt(\mathds k)$.
\end{proof}

We now consider the case where $\mathds k = \Z/2$. Suppose $(V,q)$ is a nondegenerate quadratic space over $\Z/2$, then $V$ must be even dimensional because the bilinear form $b$ is also skew-symmetric in this case. Suppose that $V$ is anisotropic, then $q(x) = 1$ for all $x \in V$. Furthermore for $x,y \in V$ we have 
\begin{equation}
    b(x,y) = q(x+y)-q(x)-q(y)=1\,.
\end{equation}
But if $x,y,z\in V$ are linearly independent then $b(x,y+z) = 1$ but also $b(x,y+z)=b(x,y)+b(x,z)=0$. Therefore any anisotropic quadratic space must be two dimensional. There is such a space $(V,q)$ where $V = \Z/2 \oplus \Z/2$ and $q(a,b) = a^2+ab+b^2$, and this determines a nontrivial element in $\Witt(\Z/2) = \Z/2$. To any nondegenerate quadratic space $(V,q)$ over $\Z/2$, the associated element in $\Witt(\Z/2) = \Z/2$ is called the Arf-invariant.



\section{The Classification from Algebraic $\rL$-theory}\label{section:Ltheory}

The goal of this section is precisely to understand the image of the map  $\Mod^2(\TS^{n-2})^\times \to (\TS^n)^\times$ in the  case where $\TS$ is the target category for Pauli stabilizer codes, and thus $(\TS)^\times$ is the spectrum of Clifford QCA. 

Our main result is to show that the mobile excitations of a fully mobile Pauli stabilizer code give a Poincaré object in the category of perfect chain complexes over $R =\Z/p[x_1^{\pm}, \dots, x_{n-2}^{\pm}]$ with quadratic functor. For this reason, they can be classified by surgery. Identifying these properties for stabilizer codes may be viewed as a natural analog to the classical conditions appearing in the surgery theory of manifolds. These conditions will manifest again in \S\ref{section:framedTO}, where the goal is to apply surgery to TQFTs rather than lattice models.

\subsection{Poincar\'e duality of Stabilizer Codes}\label{subsection:homological}
For the reader's convenience, we give a quick summary of the definitions in \S\ref{subsection:introtopauli}.
 A stabilizer code consists of a free $R$-modules $P$ equipped with a perfect sesquilinear form $\lambda: P \xrightarrow{\sim} \Hom_R(\overline{P}, R): a \mapsto \Omega(\cdot, a)$ and an $R$-submodule $L\subset P$ satisfying the Lagrangian property $L^\perp=L.$ 
Recall the construction of the following chain complex
\begin{equation*}
F_\bullet :    \cdots \longrightarrow F_2 \longrightarrow F_1 \xrightarrow{\sigma} P \xrightarrow{\delta} F_1^*\longrightarrow F_2^* \longrightarrow \cdots,
\end{equation*}
in Equation \eqref{eq:dualF}.


\begin{definition}
    We say a stabilizer code is fully mobile if  $\Ext^i_R(\overline{P/L}, R)$ has Krull dimension less or equal to zero for each $i$. Equivalently, $\Ext^i_R(\overline{P/L}, R)$ has finite cardinality for each $i\geq 1$.
\end{definition}

Let $\hat{R}=\Z/p[[x_1^{\pm}, \dots, x_{n-2}^{\pm}]]$ the $R$-module\footnote{It is not a ring as multiplication is not well-defined.} of bi-infinite formal series. It is an injective cogenerator in the category of $R$-modules. An \( R \)-module \( E \) is called an injective cogenerator (in the category of $R$-modules) if it is an injective $R$-module and for every nonzero \( R \)-module \( M \), there exists a nonzero homomorphism \( M \to E \).

The following result was first noted in \cite{RY} but follows from standard homological algebra.
\begin{proposition}
    \begin{equation}
    \mathrm{Tor}_{i}^R (P/L,\widehat R) \cong \Hom_R (\Ext^i_R(\overline{P/L}, R), \widehat R) \cong \Hom_{\Z}(\Ext^i_R(\overline{P/L}, R), \mathbb Z_n).
\end{equation} 
This give a perfect pairing between $\mathrm{Tor}_{i}^R (P/L,\widehat R)$ and $\Ext^i_R(\overline{P/L}, R)$ as abelian groups. 
\end{proposition}

\begin{theorem}\label{thm:perfectpairing}
    For a fully mobile stabilizer code, there is a natural isomorphism between $\mathrm{Tor}_{i}^R (P/L,\widehat R)$ and $\Ext^{n-2-i}_R(\overline{P/L}, R)$ for $1\leq i\leq n-3$. In other words, there is perfect pairing between $\Ext^i_R(\overline{P/L}, R)$ and $\Ext^{n-2-i}_R(\overline{P/L}, R).$
\end{theorem}
\begin{proof}
    By Lemma~\ref{lem: flat}, there exists a length $n-2$ flat resolution of $\widehat{R}$ as follows
    $$T_\bullet : T_{n-2}(:=R)\rightarrow T_{n-3}\rightarrow \cdots \rightarrow T_0\rightarrow \widehat R.$$ Moreover, for $0\leq k< n-2$, $T_k\otimes_R M=0$ for any $M$ with Krull dimension $0$. 
    
    Take $F_\bullet \otimes_R T_\bullet:$
    \[
\begin{tikzcd}[column sep=1.2em, row sep=large]
\cdots \arrow[r] 
  & F_2\otimes_R R \arrow[r] \arrow[d]
  & F_1\otimes_R R \arrow[r,"\sigma\otimes 1"] \arrow[d]
  & P\otimes_R R \arrow[r,"\delta\otimes 1"] \arrow[d]
  & F_1^*\otimes_R R \arrow[r] \arrow[d]
  & F_2^*\otimes_R R \arrow[r] \arrow[d]
  & \cdots \\
\cdots \arrow[r]
  & F_2\otimes_R T_{n-3} \arrow[r] \arrow[d]
  & F_1\otimes_R T_{n-3} \arrow[r,"\sigma\otimes 1"] \arrow[d]
  & P\otimes_R T_{n-3} \arrow[r,"\delta\otimes 1"] \arrow[d]
  & F_1^*\otimes_R T_{n-3} \arrow[r] \arrow[d]
  & F_2^*\otimes_R T_{n-3} \arrow[r] \arrow[d]
  & \cdots \\
& \vdots \arrow[d] & \vdots \arrow[d] & \vdots \arrow[d] & \vdots \arrow[d] & \vdots \arrow[d] \\
\cdots \arrow[r]
  & F_2\otimes_R T_0 \arrow[r] \arrow[d]
  & F_1\otimes_R T_0 \arrow[r,"\sigma\otimes 1"] \arrow[d]
  & P\otimes_R T_0 \arrow[r,"\delta\otimes 1"] \arrow[d]
  & F_1^*\otimes_R T_0 \arrow[r] \arrow[d]
  & F_2^*\otimes_R T_0 \arrow[r] \arrow[d]
  & \cdots \\
\cdots \arrow[r]
  & F_2\otimes_R \widehat R \arrow[r]
  & F_1\otimes_R \widehat R \arrow[r,"\sigma\otimes 1"]
  & P\otimes_R \widehat R \arrow[r,"\delta\otimes 1"]
  & F_1^*\otimes_R \widehat R \arrow[r]
  & F_2^*\otimes_R \widehat R \arrow[r]
  & \cdots
\end{tikzcd}
\]
The columns are exact because they are $T_\bullet$ tensored with a free module. Except the first and last ones, the rows are also exact because for $0\leq k< n-2$, $T_k\otimes_R \Ext^i_R(\overline{P/L}, R)=0$ due to full mobility. The top row is $F_\bullet$ while the bottom row has cohomology groups $\mathrm{Tor}_{i}^R (P/L,\widehat R)$ to the left of $P\otimes_R \widehat{R}$.  Some diagram chasing gives the desired isomorphism. 
\end{proof}

\begin{lemma} \label{lem: flat}
Let \(R = \mathds{k}[x_1^{\pm1}, \dots, x_m^{\pm1}]\), and let \(\widehat{R} = \mathds{k}[[x_1^{\pm1}, \dots, x_m^{\pm1}]]\) be viewed as an \(R\)-module. Then \(\widehat{R}\) admits a flat resolution of length \(m\),
\[
T_\bullet : T_m \, (:= R) \longrightarrow T_{m-1} \longrightarrow \cdots \longrightarrow T_0 \longrightarrow \widehat{R} \longrightarrow 0.
\]
If $\mathds{k}=\Z/n$, for each \(0 \le j < m\), one has \(T_j \otimes_R M = 0\) for every \(R\)-module \(M\) of Krull dimension \(0\).
\end{lemma}
\begin{proof}
    Let \[
T_j
=
\bigoplus_{\substack{I\subseteq \{1,\dots,m\}\\ |I|=j}}
\ \ \bigoplus_{\varepsilon\in\{\pm\}^{\{1,\dots,m\}\setminus I}}
R_{I,\varepsilon},
\]
where
\[
R_{I,\varepsilon}
:=
\mathds{k}\Bigl[\Bigl[x_i^{\varepsilon_i} \,\bigm|\, i\notin I\Bigr]\Bigr]
\Bigl[x_i^{-\varepsilon_i} \,\bigm|\, i\notin I,\ x_j^{\pm1} \,\bigm|\, j\in I\Bigr].
\] The differentials \(T_j \to T_{j-1}\) are given by alternating sums of the natural inclusions
\[
R_{I,\varepsilon}\hookrightarrow R_{J,\varepsilon},
\qquad J\subset I,\quad j=|J|=|I|-1,
\]
with signs determined by the usual \v{C}ech convention. 

To help the reader understand this construction, let us spell out the case of \(T_0\). We have
\[
T_0=\bigoplus_{\varepsilon\in\{\pm\}^m}
\mathds{k}\Bigl[\Bigl[x_1^{\varepsilon_1},\dots,x_m^{\varepsilon_m}\Bigr]\Bigr]
\Bigl[x_1^{-\varepsilon_1},\dots,x_m^{-\varepsilon_m}\Bigr].
\]
Each summand consists of formal Laurent series whose support is unbounded only in the orthant of \(\mathbf{Z}^m\) determined by the sign vector \(\varepsilon\). Equivalently, in the \(\varepsilon\)-summand one allows arbitrary infinite expansion in the directions \(x_i^{\varepsilon_i}\), while only finitely many powers occur in the opposite directions \(x_i^{-\varepsilon_i}\). The augmentation map
\[
T_0\longrightarrow \widehat{S}
\]
is given by summing the Laurent series in the various orthants, viewing each summand as a submodule of \(\widehat{S}\).

We now verify that each \(T_j\) is a flat \(R\)-module. Since \(T_j\) is a finite direct sum of the modules \(R_{I,\varepsilon}\), it is enough to show that each \(R_{I,\varepsilon}\) is flat over \(R\). Furthermore, by changing variables, they are isomorphic to \(R_{I,\varepsilon}\) for $\varepsilon\equiv 1$. 

Fix \(I\subseteq \{1,\dots,m\}\) and denote \(R_{I,1}\) by \(R_{I}\). Let
\[
A_I
:=
\mathds{k}\Bigl[x_i \,\bigm|\, i\notin I\Bigr],
\]
and let
\[
J_{I}:=\Bigl(x_i\mid i\notin I\Bigr)\subset A_I.
\]
Then
\[
\widehat{A}_I^{\,J_{I}}=
\mathds{k}\Bigl[\Bigl[x_i \,\bigm|\, i\notin I\Bigr]\Bigr],
\]
is the \(J_{I}\)-adic completion of \(A_I\). Since \(A_I\) is Noetherian, its \(J_{I}\)-adic completion is flat over \(A_I\). Adjoining the rest of the variable preserves flatness. Thus,
\[
\mathds{k}\Bigl[\Bigl[x_i \,\bigm|\, i\notin I\Bigr]\Bigr]\Bigl[x_i \,\bigm|\, i\in I\Bigr]
\]
is a flat over
\[
A_I\Bigl[x_i \,\bigm|\, i\in I\Bigr]
=
\mathds{k}[x_1,\dots,x_m].
\]
 Upon localization, it follows that \(R_I\) is a flat \(R\)-module.

Finally, when \(\mathds{k}=\mathbf{Z}/n\), Proposition~36 of~\cite{ruba2024homological} implies that
\[
T_j\otimes_R M=0
\]
for every \(R\)-module \(M\) of Krull dimension \(0\). Indeed, $M$ would be annihilated by $x_i^{\pm l}-1$ for all $i$ and any large enough $l$. However, $x_i^{\varepsilon l}-1$ is invertible in $R_{I, \varepsilon}$ for $i\notin I$ by a geometric series.
\end{proof}
\begin{rem}
    The proof resembles  \cite[Proposition 44]{ruba2024homological}. The flat resolution used here originated in joint work in progress with Błażej Ruba~\cite{RY_new}.
\end{rem}

\subsection{Coarse-Graining Applied to Stabilizer Codes}\label{subsection:coarse}

Coarse-graining is a standard tool in the study of lattice topological phases, allowing one to access their large-scale structure while systematically suppressing microscopic degrees of freedom. The purpose of this section is to formalize this notion, as it defines an equivalence relation with respect to which we classify Pauli stabilizers.

\begin{definition}
    Let $i:\Gamma\hookrightarrow \Z^m$ be a finite index subgroup. The \textit{coarse-graining} of a  $\Z/p[\Z^m]=\Z/p[x^\pm_1,x^\pm_2,\ldots,x^\pm_m]$-module $M$ with respect to $\Gamma$ is given by $i^* M$, which is $M$ treated as a $\Z/p[\Gamma]$-module spanned by the monomials $x_\gamma$ with $\gamma \in \Gamma$.
\end{definition}

Since $\Gamma$ is required to be finite index, there is an isomorphism $\Gamma\cong \Z^m$. Therefore $\Z/p[\Gamma]$ is non-canonically isomorphism to $\Z/p[\Z^m]$. Since $\Z/p[\Z^m]$ is a finitely generated free module over $\Z/p[\Gamma]$, freeness and finite generation of a module are both preserved by coarse-graining. 

\begin{proposition}[{\cite[Proposition 34]{ruba2024homological}}]
 Let $\Gamma \subset \Z^m$ be a finite index subgroup, and let $R_\Gamma$ and $R_{\Z^m}$ be two rings of Laurent polynomials defined with respect to $\Gamma$ and $\Z^m$.  Coarse-graining satisfies $\Ext^{i+1}_{R_{\Gamma}}(\overline{{i^*P}/{i^*L}},R_\Gamma) \cong i^*\Ext^{i+1}_R(\overline{P/L},R_{\Z^m})$.
\end{proposition}

 \begin{rem}\label{rem:preservedbyCG}
      This implies that if $L'\subset P'$ arises from $L\subset P$ by coarse-graining, then $\cE^i(\fC')$ is constructed from $\cE^i(\fC)$ by coarse graining, using the notation of Equation \eqref{eq:charge}. In this sense, the construction of $\cE^i$ commutes with coarse graining. Furthermore, any bilinear form defined on $P$ is also preserved by coarse graining \cite[Remark 29]{RY}.

 \end{rem}

We now present an important relation on $\rL$-theory, and describe how it behaves under the coarse-graining procedure. This analysis will play a key role in the next section, where it is used in the classification.
\begin{proposition}[{\cite[Proposition 17.3]{Ranicki1992}}]\label{prop:coarseL}
    For any $m\geq 1$ the quadratic $\rL$-groups defined over the ring $R = \Z/p[x^\pm_1,\ldots, x^\pm_m]$ satisfy the binomial formula
    \begin{equation}\label{eq:Lbinomial}
        \rL_n(\Z/p[x^\pm_1,\ldots, x^\pm_m]) = \sum^m_{i=0} \binom{m}{i}\rL_{n-i}(\Z/p)\,.
    \end{equation}
\end{proposition}
Although this group is somewhat unwieldy, coarse-graining offers a natural simplification because the coarse-graining operation can be interpreted as a colimit over finite index subgroups of $\Z^m$:
\begin{equation}
 \varinjlim_{H \leq \Z^m} \rL_n(\Z/p[x^\pm_1,\ldots, x^\pm_m])\,.
\end{equation}

\begin{proposition}[{\cite[Theorem 6.4]{haah2025topological}}]\label{prop:Haah}
    Under coarse-graining only the first term of the binomial in Equation \eqref{eq:Lbinomial} survives, i.e. 
    \begin{equation}
        \varinjlim_{H \leq \Z^m} \rL_n(\Z/p[x^\pm_1,\ldots, x^\pm_m]) = \rL_n(\Z/p)\,.
    \end{equation}
\end{proposition}

\begin{rem}\label{rem:useCG}
By combining the coarse-graining procedure with the computation of $\rL$-groups over fields described in \S\ref{subsection:Lgroupfork}, we obtain a simplified description of the $\rL$-groups that classify mobile excitations. In particular, we invoke Propositions \ref{prop:coarseL} and \ref{prop:Haah} to simplify from  $\rL_*(R)$ to $\rL_*(\Z/p)$, upon coarse-graining. Due to the fact that mobile excitations and bilinear forms are preserved under coarse graining, we can be sure that the simplified group can be used to classify them.
\end{rem}
\subsection{Surgery on Pauli Stabilizer Codes}\label{subsection:surgery}

In surgery theory, a central question is whether one can systematically eliminate the nontrivial topology of a manifold through successive surgeries, ultimately reducing it to a sphere. The algebraic surgery applied to the topological operators of a Pauli stabilizer code follows a closely analogous philosophy. In this correspondence, a boundary condition plays the role of a filling of the manifold, i.e. a null-cobordism, while an interface is interpreted as a cobordism. We now exploit this analogy to classify Pauli stabilizer codes up to surgery and coarse-graining. Before beginning, observe that in our conventions, a $(n-1)$d Pauli stabilizer code has topological operators that go up to $(n-2)$-dimensions. Theorem \ref{thm:perfectpairing} shows that the mobile excitations constitute a Poincaré object in the category $(D^{\perf}(R), \Sigma^{n-2}Q)$.



 A key feature that we must include in the axiomatization of Pauli stabilizer codes is that  operators detect each other via intersection/linking.  This is also called remote detectability.  Thus, Pauli stabilizer codes are determined not only by their mobile operators up to surgery, but also by the statistics that the middle dimensional operators after surgery exhibit with respect to intersection/linking pairing with each other in the ambient dimension of the theory; see the review in \S\ref{subsection:gs}. As an example, take $n=4$. The (2+1)d Pauli stabilizers are determined by a symmetric bilinear form \cite{RY}, which corresponds to the fact that the standard $S$-matrix encoding linking of line operators is symmetric in this dimension. However, $\rL$-theory alone on the mobile operators would suggest we look at a $\rL$-group in degree 2 mod 4, which corresponds to antisymmetric forms. More generally,  a $p$-dimensional operator (thought of as a middle dimensional operator after surgery) has antisymmetric linking with itself in $2p+1$ ambient dimensions if $p$ is even, and symmetric linking if $p$ is odd \cite[Proposition 49]{ruba2024homological}. Combining this observation with the classification of mobile excitations up to surgery we make the following prescription for the classification.

\begin{ansatz}
     Pauli stabilizer codes are classified by the following two step procedure. The first step is to apply surgery to its topological operators, which can be done thanks to Theorem \ref{thm:perfectpairing}. Then, classify the remaining operator in middle degree by the symmetric or skew-symmetric nature that it enjoys under linking given by the formula  in \cite[Proposition 49]{ruba2024homological}. 
\end{ansatz}

We find that Pauli stabilizer codes in $(n-1)$d are classified by the same group that classifies mobile operators up to surgery in 
$(n+2+4k)$d Pauli stabilizer codes.\footnote{For operators in the 
$(n+2)$-dimensional Pauli stabilizer code, we focus on the symmetric or antisymmetric nature of the pairing that controls the linking of middle-dimensional operators, rather than on the exact degree of that dimension, as the latter is not essential.}
\begin{theorem}\label{thm:PauliClassification}
     Pauli stabilizer codes in $(n-1)$-dimensions where $n>4$ are classified up to gapped interface and coarse-graining by the groups 
     \begin{equation}\label{eq:assemble}
         \rL_{-n}(\Z/p) =  \begin{cases}
  0 & n\equiv 1,3 \mod 4 \\
  \Z/2  & n \equiv 2 \mod 4 \\
  \Witt(\Z/p) & n\equiv 0 \mod 4\,.
\end{cases}
     \end{equation}
\end{theorem}

\begin{proof}
   We know that the mobile excitations which determine the Pauli stabilizer code in $(n+2+4k)$d form a Poincaré object of $(D^{\perf}(R), \Sigma^{n+4k}Q)$ thanks to Theorem \ref{thm:perfectpairing}.
   In the case when $n$ is odd, the classification of Poincaré object is by the odd $\rL$-groups. When $n=2m$ then  the middle dimension operators are in degree $j=m$ by Corollary \ref{cor:singledegree}. In the case where $m$ is odd, the classification is therefore governed by skew-symmetric forms over $R$, corresponding to $\rL$-groups in degree 2 mod 4. 
   When $m$ is even the classification is therefore governed by symmetric forms over $R$ corresponding to $\rL$-groups in degree 0 mod 4. To finish, we invoke Remark \ref{rem:useCG} to simplify the $\rL$-group that classifies mobile excitations from  $\rL_*(R)$ to $\rL_*(\Z/p)$, upon coarse-graining. The result follows from Propositions \ref{prop:4periodic}, \ref{prop:oddLgroups}, \ref{prop:2mod4} and \ref{prop:0mod4} that compute the values of $\rL$-groups over $\Z/p$ for $p=2$ and $p\neq 2$. Equation \eqref{eq:assemble} follows from assembling the results for all primes.
\end{proof}

\begin{rem}
   The type of $\rL$-theory we have been working goes by the name of quadratic $
\rL$-theory. There is also a notion of \textit{symmetric} $\rL$-theory, but away from $p=2$ these two theories give the same homotopy groups. At $p=2$ they are the same in degrees $\leq -3$. Thus, even without providing a quadratic refinement of the bilinear form we can still give a unique classification of Pauli stabilizer codes.
\end{rem}

\begin{corollary}\label{cor:nogapped}
    The classification of Clifford QCAs matches that of Pauli stabilizer codes in every dimension.
\end{corollary}
\begin{proof}
    Compare the classification in \Cref{thm:PauliClassification} with that of \Cref{thm:ClassifyCliffordQCA}, the result of which is displayed in \Cref{eq:valuesofK}. Using the property of the universal target $\TS$ applied to the category of Pauli stabilizer codes such that $(\TS)^\times$ is the spectrum of Clifford QCA, we see that the map  $\Mod^2(\TS^{n-2})^\times \to (\TS^n)^\times$ is the identity in each degree.
\end{proof}
Another way to interpret this result is that after coarse-graining, Pauli stabilizers up to gapped interface can have more interesting invariants corresponding to  ``Witt-nontrivial'' theories, and consequently are only defined in the presence of a bulk. The physical reasonability of this will be discussed in \S\ref{subsection:latticeandcontinuum}  when we compare with the classification of framed topological field theories.

\begin{rem}
    In the case when a Pauli stabilizer code is thought of intrinsically as a $n$d theory i.e. the bulk is trivial, then it can only be realized by a class that is Witt trivial. This is the case e.g. for the toric code.
\end{rem}

\section{Comparing Pauli Stabilizers to Framed Topological Field Theory}\label{section:framedTO}
We now explicate the close relationship between the classification of topological Pauli stabilizer codes and framed topological field theories, already hinted at in \S\ref{subsection:targetstabilizer}. This section serves to bridge the two communities by showing how the concepts used in the classification of both types of theories are inherently the same. Moreover, the close relationship allows for certain framed topological order to provide a continuum realization of Pauli stabilizer codes.

\subsection{Preliminaries on TQFTs}\label{subsection:prelim}
 At this stage it is important to clarify what we mean by a ``TQFT''. The cobordism hypothesis asserts that fully extended topological quantum field theories are classified by fully dualizable objects in a suitable symmetric monoidal higher category. While this perspective is mathematically powerful and conceptually far-reaching, it is broader than the framework one should adopt for physics. In particular, only imposing full dualizability on objects is not enough to construct TQFTs that are relevant to our discussion, and certainly not when it comes to comparing with topological theories on the lattice arising from stabilizer codes. Our usage of ``TQFT'' will therefore be more restrictive, tailored to capture theories called \textit{topological order}.  In essence these are TQFTs in the sense of the cobordism hypothesis with extra properties. We provide a working definition using the framework of fusion categories, for a discussion of the physical validity of the definition we refer the readers to \cite[Section III.A]{JF}.

\begin{definition}[{\cite[Definition I.1]{JF}}]\label{def:TO}
    A topological quantum field theory in $(n+1)$-dimensions is a multifusion $n$-category with trivial center.
\end{definition}

\noindent This definition recasts a physical system in precise mathematical terms. To be more explicit, we see that:
\begin{itemize}
    \item The condition of multifusion means that the theory has a finite number of operators, and there can be multiple vacua. All the operators are fully topological, meaning they have enough dualizability.
    \item Stacking of TQFT given by tensor product of multifusion  $n$-categories.
    \item Trivial center enforces remote detectability.
\end{itemize}
We will mainly contend with the case when the category is just fusion, meaning that the identity is a simple object. However, now the task is to understand how to construct higher fusion categories in a systematic way and justify that they enjoy the properties one expects in order to be relevant for TQFTs. A key structural ingredient in their construction is \textit{condensation completion}, which generalizes the familiar notion of Karoubi completion, i.e. splitting of idempotents, to the higher-categorical setting. The importance of idempotent completion can already be seen in 1-categories, with application to (2+1)d TQFTs. In particular, if $X$ is a line operator, then if $X$ is not a simple object, that implies $\End(X)\neq \mathbb C$. Given a projector $P\in \End(X)$ which is idempotent, to have $P$ split means that we can realize the image of $P$ as a direct summand of $X$.

Condensation completion is used to ensure the appropriate splitting of higher idempotent structures and thereby produces the kinds of semisimple objects expected in the description of topological phases. 

\begin{definition}[{\cite[Definition II.3]{JF}}]
    Let $\mathbf{C}$ be an $n$-category with  objects $X,Y \in \mathbf{C}$. A \textit{condensation} of $X$ onto $Y$ consists of morphisms $f:X \rightleftarrows Y:g$ together with a condensation $fg$ onto $\id_Y$. But the latter is a condensation in the $(n-1)$-category $\End_{\mathbf{C}}(Y)$, so the definition of condensation is given inductively. The induction begins by declaring that equalities are examples of condensations.
    \end{definition}

    The notion of an idempotent is replaced by a condensation monad \cite{Gaiotto:2019xmp}: the axiom $P^2=P$ for projector is replaced by a condensation of $P^2$ onto $P$. The $n$-category $\mathbf{C}$ is condensation complete if every condensation monad factors through a condensation. The construction in \cite[Section 2.4]{Gaiotto:2019xmp} implies that if $\mathbf{A}$ is the monoidal $n$-category of operators in a $(n+1)$d TQFT, then $\mathbf{A}$ is condensation complete.
Having defined condensation completion we can now discuss how to build higher semisimple categories, and in this way also higher fusion categories by considering the fully dualizable objects in the Morita category of condensation complete $n$-categories.
\begin{definition}
     Given any multifusion $(n-1)$-category $\mathbf{A}$, we let $\mathrm{B}\mathbf{A}$ be the locally semisimple $n$-category with a single object $*$ and $\End_{\mathrm{B}\mathbf{A}}(*)=\mathbf{A}$. Then, we define $$\Mod(\mathbf{A}):= \mathrm{Kar}(\mathrm{B}\mathbf{A})$$ as the finite semisimple $n$-category obtained by taking the Karoubi/condensation completion of $\mathrm{B}\mathbf{A}$. Furthermore we have 
\begin{equation}
\Omega(\Mod(\mathbf{A}))\simeq\mathbf{A} \,.\footnote{If a $n$-category $\mathbf{C}$ is connected, i.e.\ that $\pi_0(\mathbf{C})$ is a singleton, we also have
\begin{equation}\label{e1}
\Mod(\Omega\mathbf{C})\simeq\mathbf{C}\,.
\end{equation}}
\end{equation}

\end{definition}

\begin{rem}
   For small values of $n$, existing results \cite[Section 3.2]{Bhardwaj:2024xcx} imply that the $n$-category $\Mod(\mathbf{A})$ is finite semisimple. Thus it behaves as a semisimple category of modules in the relevant higher-categorical sense. When $n >3$, semisimplicity is subtler to prove and does not follow automatically from the construction of \cite{GJF19}. In what follows we therefore treat finite semisimplicity of 
$\Mod(\mathbf{A})$ as a working hypothesis, which is consistent with all examples we consider and is expected to admit a general proof. We leave the full higher-categorical proof to a dedicated treatment elsewhere.
\end{rem}

Physically, condensation completion should be interpreted as a fact that in TQFTs we should not treat certain trivial operators only as non-genuine operators, rather than totally trivial. They are non-genuine in the sense that they arise as \textit{condensation descendants} from operators of lower dimensions. The most salient feature of a category (of operators) $\mathcal{U}$ that is a condensation descendant is that it always admits as a boundary given by a category of objects in lower dimension $\cV$; categorically this realizes $\cU$ as $\Mod(\cV)$. This is the precise sense by which we mean a theory is built from a theory of lower dimension, i.e. it is realized as $\Mod(\cC)$ and admits a boundary given by $\cC$ which has its own category of topological operators.

\begin{example}
    In a specific example, we take $\ \Mod (\mathbb C) =\Vect$, the category of finite dimensional vector spaces. This category  has only one simple object, which we regard as the 2-dimensional vacuum theory. We can also do this iteratively $\Mod^n(\mathbb C) = \mathbf{nVect}$ to obtain the $(n+1)$-dimensional vacuum. We see that indeed the vacuum is solely comprised of non-genuine operators: the operators of higher dimensions can all end on an operator in one lower dimension, all the way down to point operators. This implies that a Morita equivalence to $\mathbf{pVect}$ will enforce that operators in dimension $p$ and below are condensation descendants.  
\end{example}

\begin{rem}
Even though we have seen many similarities between the theory of Pauli stabilizers codes and TQFTs, the category of excitations of the former does not usually include ones that arise from condensation completion. Instead, the operators are just the mobile excitations. More needs to be done to obtain the condensation operators, see \cite{Chen:2023qst} for some constructions.
\end{rem}

\subsection{Classifying Framed Topological Field Theories}\label{subsection:classifyingTO}
The classification of topological field theories was first strategized by Lan-Kong-Wen as a dual to surgery of manifolds \cite{Lan_2018,Lan_2019}. The concepts therein were further developed in \cite{JF}, and utilized in \cite{Decoppet:2025eic,JFY2} for classifying TQFTs in (3+1)d and (4+1)d. We will be mainly concerned with framed TQFTs in this section. A framed TQFT can be thought of as the low energy limit of a phase of matter in which one is allowed to  couple a framing of the spacetime lattice, on which defines the theory. The surgery theory of compact manifolds relies fundamentally on Poincaré duality. An analogous mechanism appears in the study of topological order. In this latter case, the role of Poincaré duality is played by the triviality of the center in Definition \ref{def:TO}, which enforces a corresponding pairing between excitations in complementary dimensions. In (2+1)-dimensions, this condition can be made completely explicit, and is equivalent to the non-degeneracy of the modular $S$-matrix.

Given a TQFT described by a fusion 
$n$-category, the appropriate notion of classification which is analogous to classification up to surgery in the manifold setting, is classification up to \textit{Morita equivalence}. In this framework, two phases are regarded as equivalent precisely when they are related by an invertible bimodule, reflecting the physical idea that they differ only by the insertion of gapped interfaces.
In particular, the sequence of gapped interfaces used in \S\ref{subsection:targetstabilizer} to classify stabilizer codes admits a direct analogue in this setting: it realizes Morita equivalences between fusion 
$n$-categories describing topological orders. We now describe a procedure for constructing Morita equivalences that implement the role of surgery outlined in the previous section.

\begin{construction}\label{construction}
Let  $\mathbf{A}$ be a fusion $n$-category. Its category of dimension $p$-operators is given by $\Omega^{n-p} \mathbf{A}$.
\begin{enumerate}
    \item The category of operators $\Omega^{n-p}\mathbf{A}$, for $p=1, \ldots \lfloor \frac{n-1}{2}\rfloor$ are all given by symmetric categories.\footnote{In $n$-dimensions, the algebra of $p$-dimension operators is automatically $E_{n-p}$-monoidal. and $E_{n-p}$-monoidal is $E_{\infty}$-monoidal as soon as $n-p \geq  p+2$. This is fact is called the stabilization hypothesis, introduced by Baez-Dolan \cite{Baez:1995xq}}\\
    \item We can choose a fiber functor $F: \Omega^{n-1}\mathbf{A} \to \mathbf{Vect}$, and suspend it to a functor $\Mod^{n-1}(F): \Mod^{n-1}(\Omega^{n-1}\mathbf{A}) \to \mathbf{nVect}$. $\Mod^{n-1}(\Omega^{n-1}\mathbf{A})$ is the sub-$n$-category of $\mathbf{A}$ consisting of operators arising as condensation descendants of $1$-dimensional operators.\\
\item The functor $\Mod^{n-1}(F)$ makes $\mathbf{nVect}$ into a module for $\Mod^{n-1}(\Omega^{n-1}\mathbf{A})$. Take
the base change of this module along the inclusion $\Mod^{n-1}(\Omega^{n-1}\mathbf{A})\subset \mathbf{A}$ which produces a module $\mathbf{M}_1$ for $\mathbf A$
given by:
\begin{equation}
    \mathbf M_1 : = \mathbf{A} \boxtimes_{\Mod^{n-1}(\Omega^{n-1}\mathbf{A})} \mathbf{nVect}\,.
\end{equation}
The category $\End_{\mathbf A}(\mathbf M_1)$ has no nontrivial $1$-dimensional operators by construction, and $\mathbf M_1$ witnesses a Morita equivalence between $\mathbf{A}$ and $\End_{\mathbf A}(\mathbf M_1)$.\\
\item Iterate the first three steps: starting with $\End_{\mathbf{A}}(\mathbf{M}_1)$, we can construct the module 
\begin{equation}
   \mathbf M_2 :=  \End_{\mathbf{A}}(\mathbf{M}_1) \boxtimes_{\left(\Mod^{n-2}(  \End_{\mathbf{A}}(\mathbf{M}_1)\right))} \mathbf{nVect}\,.
\end{equation}
The category $\End_{\End_{\mathbf{A}}(\mathbf{M}_1)}(\mathbf M_2)$ has no nontrivial $1$ and $2$-dimensional operators.\\

\item We iterate to get the category $\End_{\End_{(\ldots)}(\mathbf{M}_{\lfloor \frac{n-1}{2} \rfloor-1})}(\mathbf M_{\lfloor \frac{n-1}{2} \rfloor})$
which has no operators of dimension up to and including $\lfloor \frac{n-1}{2} \rfloor$. By triviality of the center, this implies that all operators which link with the ones that have been removed using the steps above, arise as condensation descendants in the sense of \cite{Gaiotto:2019xmp}.\\
\item Classify the middle dimensional operators that remains, if there are nontrivial ones.
\end{enumerate}

\end{construction}


\begin{example}
   The technique explained in Construction \ref{construction} have been integral in solving many problems regarding (3+1)d and (4+1)d  TQFT, in both the bosonic and fermionic cases. In particular, it was shown that bosonic (3+1)d TQFTs are classified by a finite group $G$ and a class in $\rH^4(\rB G; \mathbb C^\times)$.
   This understanding of TQFTs has also contributed in a host of ways to solving problems in (2+1)d and (3+1)d. This involves:
   understanding of anomaly indicators for fermionic symmetries \cite{Debray:2023iwf}, the classification of fusion 2-categories \cite{Decoppet:2024htz}, duality symmetry in (3+1)d \cite{Bhardwaj:2024xcx,DelZotto:2025yoy}, how non-invertible symmetries act on local operators \cite{Putrov:2026csz}, the classification of symmetry enriched topological order \cite{Decoppet:2025eic}, and symmetry enforced gaplessness \cite{Debray:2025kfg,Debray:2026sqw}.
\end{example}

\begin{rem}
  Topological order and TQFTs are more naturally aligned with manifold-level phenomena: their invariants and classification reflect genuinely geometric features, analogous to the role of $\rL$-theory over $\mathbb{Z}$, which detects subtle manifold-theoretic structure such as the signatures. By contrast, stabilizer codes are governed primarily by algebraic data, closer in spirit to $\rL$-theory over $\mathbb{Z}_p$, where the structure is controlled by purely algebraic invariants without direct sensitivity to smooth or geometric refinement.
In this sense, the distinction parallels that between integral and mod-$p$ information. The former retains global geometric content, while the latter captures formally algebraic structure. 
\end{rem}


\begin{theorem}\label{thm:classificationTQFT}
    Let $n>4$. The classification of anomalous $(n-1)$d framed topological field theories valued in the universal target $\mathbf{W}$ in Remark \ref{rem:semisimpletarget}, up to Morita equivalence, is given by 
    \begin{equation}\label{eq:Lframed}
            \Mod^2(\mathbf{W}^{n-2})= \begin{cases}
  0 & n \equiv 1,3 \mod 4 \\
  \Z/2  & n \equiv 2 \mod 4 \\
  \Witt^\pt & n\equiv 0 \mod 4\,,
        \end{cases}
    \end{equation}
    where 
    \begin{equation}\label{eq:pointedWitt}
        \Witt^{\pt} = \bigoplus_{p} \begin{cases}
             \Z/2\oplus \Z/8 & p=  2 \\
  \Z/2\oplus \Z/2   & p \equiv 1 \mod 4\\
  \Z/4 & p\equiv 3 \mod 4\,,
        \end{cases}
    \end{equation}
\end{theorem}
\begin{proof}
   We use the steps in Construction \ref{construction} to establish the Morita equivalences. A $(n-1)$ framed TQFT is described by a fusion $(n-2)$-category, and so we only consider the operator up to dimension $\lfloor \frac{n-3}{2} \rfloor$ are symmetric. We note that if $n$ is odd, then $\lfloor \frac{n-3}{2} \rfloor = \lfloor \frac{n-2}{2} \rfloor$. Therefore by condensing all the symmetric operators, we are actually left with no middle dimensional operator to classify. Therefore in the case when $n$ is odd, all $(n-1)$d TQFTs are Morita equivalent to an invertible TQFT and this gives the first case in \Cref{eq:Lframed}. 
   
   If $n=2m$ then there are middle dimensional operators that is not symmetric monoidal. However, it follows from \cite{JFY1} that a $m$-category for $m\geq 2$ with only genuine operator in dimension $m$, has fusion algebra given by a twisted group algebra of a finite group $A$. The group algebra is required to be $\mathbf{W}$-linear, and hence the twist is a cohomology class valued in $I_{\mathbb{C}^\times}$. The middle dimensional operators are moreover $E_m$-monoidal and thus the twisted group algebra is described by a class in $I^{2m}_{\mathbb{C}^\times}(\rB^m A)$, which factors through $I^{2m}_{\mathbb{Q}/\Z}(\rB^m A)$ because $A$ is finite. A result of \cite{JFR2025} shows that $I^{2m}_{\mathbb{C}^\times}(\rB^m A)$ is the set of either symmetric (when $m$ is even) or antisymmetric (when $m$ is odd) $\mathbb Q/\Z$-valued  bilinear forms on $A$. The classification of TQFTs then reduces to a classification of such forms up to isotropic reduction, i.e. the Witt group of finite abelian groups when $m$ is even and $\Z/2$ when $m$ is odd.
\end{proof}

\begin{rem}
    The previous theorem has applications in the work of Johnson-Freyd-Reutter, concerning computing the properties of the universal target $\mathbf{W}$. In particular there is a short exact sequence \cite{JFtalk1,JFtalk2,JFtalk3,JFSDVS,JFtalk5:KITP}, that can be used to compute the higher absolute Galois groups of $\mathbb R$. The analogue of the short exact sequence is something we do not know for stabilizer codes. 
\end{rem}

\subsection{Relations Between the Lattice and Continuum}\label{subsection:latticeandcontinuum}

We end with a comparison between the nature of invertible stabilizer codes, described in Corollary \ref{cor:nogapped}, and invertible framed field theories in the continuum. This will require the following theorem due to Johnson-Freyd-Reutter.
\begin{theorem}[\cite{JFR2025}]
    Except for six exception,  corresponding to the dimensions where there exists a nontrivial Arf-Brown-Kervaire invariant, a nontrivial invertible framed TQFT does not admit a topological boundary condition in any semisimple category.
\end{theorem}

One way to summarize this result is by studying the image of the map $\Mod^2(\mathbf W^{n-2}) \to I^n_{\mathbb C^\times}(\pt)$. When $n \equiv 2 \mod 4$ the map produces a TQFT, which we can think of as a generalized Dijkgraaf-Witten theory, that is the Arf-Brown-Kervaire invariant.

\begin{corollary}
    When $n \equiv 0 \mod 4$ the map $\Mod^2(\mathbf W^{n-2}) \to I^n_{\mathbb C^\times}(\pt)$ is trivial.
\end{corollary}
\begin{proof}
    Any symmetric $\mathbb Q/\Z$-valued form on an abelian group admits a quadratic refinement. The choice of quadratic refinement selects an orientation structure on the framed invertible theory.\footnote{For example, in (2+1)d a ribbon structure allows for one to define a self twist of an anyon. This is a quadratic refinement of the braiding. } Any oriented invertible theory in dimension $n \equiv 0 \mod 4$ is classified by the oriented bordism group, with the bordism invariants given by characteristic classes for the tangent bundle. However, all characteristic classes vanish on framed manifolds so the map is trivial. 
\end{proof}

In the absence of global symmetry, the group
$I^n_{\mathbb C^\times}(\pt)$ can be reframed to parametrize purely gravitational anomalies for the boundary $(n-1)$d theory. Therefore if a  $(n-1)$-dimensional theory has a gravitational anomaly that is not one of the six Arf-Brown-Kervaire invariants, then it is gapless.
This contrasts with the lattice framework, such as that of invertible Pauli stabilizer codes, also corresponding to Clifford QCA. Indeed, we notice that Corollary \ref{cor:nogapped}  implies that since the spectrum of Clifford QCAs has groups that are trivial in $I_{\mathbb{C}^\times}(\pt)$, they can admit gapped boundaries in more than just those dimensions $n \equiv 2 \mod 4$. Rather, they can admit gapped boundaries for all $n \equiv 0 \mod 2$ which is not the case for invertible framed TQFTs. 

\begin{rem}
   The classification in \Cref{thm:PauliClassification} resembles the classification of framed TQFTs in \Cref{thm:classificationTQFT}, with the caveat that only working with prime qudits on the lattice, will not give as rich of a decomposition into primes as $\Witt^\pt$ does in Equation \eqref{eq:pointedWitt}. However, notably, this gives strong evidence to the continuum limit of certain lattice theories under our consideration. We predict that in the cases when $n\equiv 0 \mod 4$ the invertible  bulk  theories constructed by Clifford QCAs, which admit  gapped boundaries, will exhibit the gap closing on its boundary when taking the continuum limit. 
\end{rem}

For completeness we now turn our attention to the classification of lower dimensional TQFTs.



\begin{proposition}
    In (1+1)d, TQFTs decompose into a direct sum of vacua. Over a single vacuum, there are no nontrivial TQFTs.
\end{proposition}
\begin{proof}
    Let $\cC$ be a (1+1)d TQFT. By Definition \ref{def:TO} the algebra of 0-dimensional operators, denoted $\Omega \cC$, is a commutative separable finite-dimensional algebra over $\mathbb C$. An algebra of this form decomposes as a direct sum of copies of $\mathbb C$ indexed by the set $\hom(\Omega \cC, \mathbb C)$. For a point $p \in \hom(\Omega \cC, \mathbb C)$ one obtains a projector $\delta_p \in \Omega \cC$ which can be used to project onto a direct summand $\mathds 1_p = \delta_p \mathds{1} \delta_p$ of $\mathds{1} \in \cC$. The object $\mathds 1_p$ is automatically a separable associative algebra in $\cC$, thus we can consider the category of bimodules $\Bimod_{\cC}(\mathds 1_p)$, with Morita equivalence implemented by $M_p = \Mod(\mathds{1}_p)$. An object of $M_p$ is a bimodule between $\mathds 1_p$ and the unit $\mathds 1$, which decomposes as a sum over $\hom(\Omega \cC,\mathbb{C})$. This implies a decomposition of the category $M_p$ as a direct sum of categories
    $$M_p= \bigoplus_{p'\in \hom(\Omega\cC,\mathbb C)} M_{p,p'}$$
    where $M_{p,p'}$ are the categories of $\delta_{p}\text{-}\delta_{p'}$ bimodule objects in $\cC$. Using the bulk-boundary relationship, we know that $\End_\cC(M_p) \simeq \Bimod_\cC(\mathds 1_p)$ and one can compute  that
    $$\End_\cC(M_p) = \bigoplus_{p',p''\in \hom(\Omega \cC, \mathbb C)} \hom(M_{p,p'},M_{p,p''}).$$  
    Using the Morita equivalence between $\cC$ and $\Bimod_{\cC}(\mathds 1_{p})$, we can then compute that
    \begin{equation}
        \Omega \cC = \bigoplus_{p' \in \hom(\Omega \cC,\mathbb C)} \hom(M_{p,p'},M_{p,p'}).
    \end{equation}
    By the fact that the algebra describing 0-dimensional operators decomposes as a direct sum of copies of $\mathbb{C}$, we see that each summand $M_{p,p'}$ is (Morita) invertible in the category of Karoubi complete categories. But the only fusion category that is invertible is the vacuum $\mathbf{Vec}_{\mathbb{C}}$. This implies 
    \begin{equation}
        M_p = \bigoplus_{\hom(\Omega \cC,\mathbb C)} \mathbf{Vect}_{\mathbb{C}}.
    \end{equation}
    We therefore see that $\cC = \End_\cC(M_p)$ is given by a matrix category, mapping between distinct vacua. If the category only had a single vacua then there are no nontrivial morphisms and therefore there is nontrivial TQFT.
\end{proof}

  What is left is to explain the  classification of TQFTs in (2+1)d. The TQFTs are classified by their Witt class, in the Witt group introduced in \cite{davydov2010witt}. We denote this group by $\cW$. The structure of this group can be gleaned from another group called the Witt group of \textit{slightly degenerate} braided fusion categories, denoted $\sWitt$. The full structure of this group is given in \cite{davydov2013structure}:
\begin{equation}
    \sWitt = \sWitt^{\pt} \oplus \sWitt_{2} \oplus \sWitt_{\infty}
\end{equation}
  where $ \sWitt^{\pt}$ is generated by the Witt classes of slightly degenerate pointed braided fusion categories, $\sWitt_{2}$ is an elementary abelian 2-group shown in \cite[Theorem 7.2]{ng2022higher} to be of infinite rank, and $\sWitt_{\infty}$ is a free group of countable rank. 

  \begin{theorem}
       The map $\cW \to \sWitt$ is surjective with a $\Z/{16}$ kernel given by the theories $\Spin(n)_1$.
  \end{theorem}
 The first part of this theorem is due to \cite{JFR} and latter due to {\cite[Proposition 5.14]{davydov2013structure}}. 
The main takeaway is that in the continuum, there are many more (2+1)d TQFTs than just those that arise from pointed categories, in sharp contrast with the more rigid situation in higher dimensions. The TQFTs that are classified by classes in the group $\sWitt_{2} \oplus \sWitt_{\infty}$ are more naturally associated with categories that come from representations of VOAs. For example, a distinguished subset of generators of this group is given by the modular tensor categories $\cC(\mathfrak g,k)$ of  highest weight integrable modules of an affine Lie algebra $\hat{\mathfrak{g}}$ at level $k$. The TQFTs constructed from such categories can admit gapless (1+1)d boundary theories, and to the best of our knowledge, no lattice realization of these theories is currently known.
In comparing with what can be constructed on the lattice, it seems that Pauli stabilizers codes in (2+1)d can only access classes in the pointed part of the group $\cW$, see \cite{Shirley:2022lhu}. Furthermore, the authors of \cite{Shirley:2022lhu} realize the (2+1)d theory as  the boundary of a bulk, rather than a bulk for a (1+1)d boundary, which is not in contradiction to the fact that e.g. we do not know how to realize Chern-Simons on a finite dimensional Hilbert space.
It was conjectured in \cite{Haah:2019fqd} that the full group of QCAs in four dimensions is the group $\cW$. Therefore whatever lies beyond the pointed part of $\cW$ also lies beyond the realm of Pauli stabilizer codes and Clifford QCA.

If indeed the classification of QCAs in (3+1)d is correct then there should exist stabilizer codes in (2+1)d, classified up to gapped interface, which correspond to nontrivial classes in $\sWitt_{2} \oplus \sWitt_{\infty}$. Drawing intuition from the continuum, we note that, unlike in higher dimensions, (2+1)d TQFTs are not classified by surgery. Notably, Construction~\ref{construction} does not produce grouplike operators: in (2+1)d TQFT there are no symmetric operators to condense, so the line operators are already ``middle dimensional''. Due to the existence of nonabelian anyons, the theorem of \cite{JFY1} fails to apply, and the classification becomes wild. We would predict that the same is true for stabilizer codes in (2+1)d.  Another conundrum is that the $\mathrm{K}$-theoretic formulation of QCAs employs matrix algebras as its algebras of observables. This naturally raises the question of how such a lattice-based framework, built from local algebras, can give rise to structures that resemble VOAs.


We now shift focus to a comparison between lattice and continuum theories via their anomalies. It has been proposed in~\cite{TLE:2025} that the generalized cohomology theory governing anomalies of lattice systems is represented by the spectrum of QCAs, $\Q(\pt)$. From our knowledge of the classification of Clifford QCAs~\cite{Haah:2019fqd, haah2025topological} we know there is a map from $\Witt(\Z/p)\to \pi_0(\Q(\Z^4))$, which is neither injective nor surjective. The reason for this is that a nontrivial Clifford QCA may equal to a non-Clifford circuit. A particular example is found in~\cite{fidkowski2024qca}. See Table I of~\cite{sun2025clifford} for more examples including in other dimensions. The spectrum that classifies anomalies for (bosonic) TQFTs up to (2+1)d is the spectrum $\mathbf{4Vect}^\times$. The category $\mathbf{4Vect}$ is the Morita category of braided fusion 1-categories, constructed in \cite[Section 2.5]{Decoppet:2024htz}. The  objects are braided fusion 1-categories, and morphisms and higher morphisms are implemented by bimodules, bimodule functors, and natural transformations of bimodule functors.

By studying only the invertible objects and morphisms, we see that the homotopy groups  of $\mathbf{4Vect}^\times$ are given by
\begin{center}
\begin{tabular}{c|c}
$\pi_*$& $\mathbf{4Vect}^\times$  \\ 
\midrule
$0$  & $\mathbb{C}^\times$  \\
$-1$ & $0$  \\
$-2$ & $0$  \\
$-3$ & $0$        \\
$-4$ & $\cW$ \,,
\end{tabular}
\end{center}
with the following descriptions:
\begin{itemize}
    \item The invertible objects are Morita classes of nondegenerate braided fusion categories by \cite{brochier2021invertible}. This is captured by the group $\cW$.
    \item  $\pi_{-3}$ is trivial because it has the description as nondegenerate braided fusion 0-categories, and the only one is $\mathbb C$ itself.
    \item  $\pi_{-2}$ is trivial because every central simple algebra over $\mathbb C$ is Morita-trivial.
    \item $\pi_{-1}$ is trivial because the only vector space that is invertible is one dimensional
    \item$\pi_0$ gives the invertible top morphisms, which are just invertible complex numbers.
\end{itemize}

  This supports the conjecture of Tu-Else-Long \cite[Equation 120]{TLE:2025} in the case when there is no extra global $G$-symmetry.
In particular, the conjecture states that the anomaly for TQFTs should inject onto the anomalies on the lattice, and by our knowledge of the low homotopy groups of the spectrum of QCA, this is true.



\section{Discussion}
Our main result shows a match between the classification of Pauli stabilizer codes and TQFTs in dimension $n\geq 4$, a fact already hinted at by the relationship studied in \cite{RY} between Pauli stabilizer codes and anyon models. We have exhibited a clear distinction between invertible phases in the lattice and continuum settings, especially regarding their capacity to admit gapped boundaries. Conceptually, this difference arises from the observation that algebraic $\rL$-theory over fields captures invariants that differ from those of smooth manifolds.

It would be of considerable interest to investigate further the structural properties of the universal target for stabilizer codes.
In particular, it is also known that $ \mathbf{W}^1 = \mathbf{sVect}$, reflecting the fact that $\mathbf{W}$ is the higher algebraic closure of $\mathbb R$. It would be desirable to determine whether 
$\TS$ admits an analogous structural property.
The universal target for semisimple categories also has two more features that are reasonably well understood. First, one has a concrete description in the situation where a category only has 0-morphisms, and all higher morphisms are trivial. Second, one has partial information about the negative homotopy groups of $I_{\mathbb C^\times}(\pt)$. These results make computations involving the universal target, and consequently the associated higher Galois groups, more tractable. On the other hand, the negative homotopy groups of the QCA spectrum remain poorly understood. Concrete computations of these groups would represent a substantial advance, both in clarifying the  structure of the universal target of stabilizer codes, but also in the study of lattice anomalies.

\section*{Acknowledgments}
It is a pleasure to thank Yuhan Gai, Theo Johnson-Freyd, Rajath Radhakrishnan, and Błażej Ruba,
for helpful discussions. We would like the thank New York University for hosting the annual pre-meeting for the Simons Collaboration on Global Categorical Symmetries, where this collaboration was initiated. BY acknowledges Agnès Beaudry, Michael Hermele, Wilbur Shirley, Evan Wickenden for discussions and collaboration. 
BY is supported by the Simons Foundation through the Simons Collaboration on Global Categorical Symmetries. MY is supported by the EPSRC Open Fellowship EP/X01276X/1.

\bibliographystyle{alpha}
\bibliography{references.bib}
\end{document}